\documentclass[twocolumn,secnumarabic,nobibnotes,superscriptaddress,aps,nofootinbib,10pt,pra]{revtex4-2}
\usepackage{color, xcolor, colortbl}
\usepackage{graphicx}
\usepackage{amsmath,amssymb,amsthm,amsfonts}
\usepackage{mathtools}
\usepackage[ruled,vlined,linesnumbered]{algorithm2e}
\usepackage{algpseudocode}
\usepackage{dcolumn}
\usepackage{epstopdf}
\usepackage{bm}
\usepackage{appendix}
\usepackage{multirow}
\usepackage{braket}
\usepackage[english]{babel}
\usepackage[T1]{fontenc}
\usepackage{xpatch}
\usepackage{tikz}
\usepackage{adjustbox}
\usepackage{xspace}
\usepackage[roman]{complexity}
\usepackage{xparse}
\usepackage{comment}
\usepackage{balance}
\usetikzlibrary{decorations.markings, arrows.meta}
\usepackage{hyperref}
\usepackage[capitalize]{cleveref}

\theoremstyle{plain}
\newtheorem{theorem}{Theorem}

\newtheorem{lemma}[theorem]{Lemma}

\newtheorem{problem}{Problem}
\theoremstyle{definition}

\theoremstyle{remark}

\crefalias{section}{appendix}

\newcommand{\eqn}[1]{\hyperref[eqn:#1]{(\ref*{eqn:#1})}}
\newcommand{\rem}[1]{\hyperref[rem:#1]{Remark~\ref*{rem:#1}}}
\newcommand{\thm}[1]{\hyperref[thm:#1]{Theorem~\ref*{thm:#1}}}
\newcommand{\cor}[1]{\hyperref[cor:#1]{Corollary~\ref*{cor:#1}}}
\newcommand{\defn}[1]{\hyperref[defn:#1]{Definition~\ref*{defn:#1}}}
\newcommand{\lem}[1]{\hyperref[lem:#1]{Lemma~\ref*{lem:#1}}}
\newcommand{\prop}[1]{\hyperref[prop:#1]{Proposition~\ref*{prop:#1}}}
\newcommand{\fig}[1]{\hyperref[fig:#1]{Figure~\ref*{fig:#1}}}
\newcommand{\tab}[1]{\hyperref[tab:#1]{Table~\ref*{tab:#1}}}
\newcommand{\algo}[1]{\hyperref[algo:#1]{Algorithm~\ref*{algo:#1}}}
\renewcommand{\sec}[1]{\hyperref[sec:#1]{Section~\ref*{sec:#1}}}
\newcommand{\append}[1]{\hyperref[append:#1]{Appendix~\ref*{append:#1}}}
\newcommand{\prob}[1]{\hyperref[prob:#1]{Problem~\ref*{prob:#1}}}
\newcommand{\assump}[1]{\hyperref[assump:#1]{Assumption~\ref*{assump:#1}}}

\newcommand{\norm}[1]{\left\lVert#1\right\rVert}

\newcommand{\Or}{\mathcal{O}}

\renewcommand{\d}{\mathrm{d}}
\newcommand{\bvec}[1]{\mathbf{#1}}
\newcommand{\spec}{\mathrm{Spec}}

\newcommand{\ve}{\bvec{e}}

\newcommand{\vp}{\bvec{p}}
\newcommand{\vq}{\bvec{q}}

\newcommand{\diag}{\operatorname{diag}}
\renewcommand{\Re}{\operatorname{Re}}
\renewcommand{\Im}{\operatorname{Im}}

\newcommand{\Tr}{\operatorname{Tr}}

\newcommand{\sinsv}{\sin^{\mathrm{SV}}}

\newcommand{\mc}[1]{\mathcal{#1}}

\newcommand{\NN}{\mathbb{N}}
\newcommand{\RR}{\mathbb{R}}
\newcommand{\CC}{\mathbb{C}}

\renewcommand{\AA}{\mathcal{A}}
\newcommand{\GG}{\mathcal{G}}
\newcommand{\HH}{\mathcal{H}}
\renewcommand{\SS}{\mathcal{S}}

\newcommand{\ketbra}[2]{\mathinner{|{#1}\rangle\!\langle{#2}|}}

\newcommand{\sbraket}[1]

\newclass{\PreciseQMA}{PreciseQMA}

\newcommand{\DeptMath}{Department of Mathematics, University of California, Berkeley, CA 94720, USA}
\newcommand{\LBLMath}{Applied Mathematics and Computational Research Division, Lawrence Berkeley National Laboratory, Berkeley, CA 94720, USA}
\newcommand{\Simons}{Simons Institute for the Theory of Computing, University of California, Berkeley, CA 94720, USA}
\newcommand{\MarylandMath}{Department of Mathematics, University of Maryland, College Park, MD 20742, USA}
\newcommand{\MarylandCS}{Department of Computer Science, University of Maryland, College Park, MD 20742, USA}
\newcommand{\MarylandQuICS}{Joint Center for Quantum Information and Computer Science, University of Maryland, College Park, MD 20742, USA}

\begin{document}

\title{Quantum Estimation of Non-Hermitian Pseudospectra}

\author{Gengzhi Yang}
\thanks{These authors contributed equally to this work.}
\affiliation{\MarylandMath}
\affiliation{\MarylandQuICS}
\author{Jiaqi Leng}
\thanks{These authors contributed equally to this work.}
\affiliation{\Simons}
\affiliation{\DeptMath}
\author{Xiaodi Wu}
\affiliation{\MarylandQuICS}
\affiliation{\MarylandCS}
\author{Lin Lin}
\thanks{linlin@math.berkeley.edu}
\affiliation{\DeptMath}
\affiliation{\LBLMath}

\date{Latest revision: \today}

\begin{abstract}
Non-Hermitian many-body systems can be spectrally unstable, so small perturbations may induce large eigenvalue shifts. The pseudospectrum quantifies this instability and provides a perturbation-robust diagnostic. For inverse-polynomially small $\epsilon$, we show that deciding whether a point $z\in\CC$ is $\epsilon$-close to the spectrum is $\PSPACE$-hard for $5$-local operators, whereas deciding whether $z$ lies in the $\epsilon$-pseudospectrum is $\QMA$-complete for $4$-local operators. This identifies pseudospectrum membership as a more realistic computational target. We then present a concrete quantum framework for deciding pseudospectrum membership, which combines a singular-value estimation step with a dissipative state preparation algorithm. Our Quantum Singular-value Gaussian-filtered Search (QSIGS) combines quantum singular value transformation (QSVT) with classical post-processing to achieve Heisenberg-limited query scaling for singular-value estimation. To prepare suitable input states, we introduce an algorithmic Lindbladian protocol for approximate ground right singular vectors and prove its effectiveness for the Hatano--Nelson model. Finally, we demonstrate the full pipeline on a trapped-ion quantum computer and distinguish points inside and outside the target pseudospectrum near the exceptional point of a minimal non-Hermitian qubit model.
\end{abstract}
\maketitle

\section{Introduction}
Non-Hermitian operators arise naturally as effective descriptions of conditional and post-selected quantum dynamics, lossy and driven systems, photonic and wave platforms, and reduced open-system dynamics~\cite{ashida2020non,xiao2024non,wiersig2020review,chaduteau2026lindbladian,almanakly2026probing,tao2025imaginary}. They also underlie a range of phenomena, including $\mathcal{PT}$-symmetry breaking, non-Hermitian skin effects, and exceptional points~\cite{yao2018edge,dora2020quantum,gopalakrishnan2021entanglement,zhang2025observation,bender1998real,bender1999PT,el2018non,zhang2022review,heiss2012physics,kawabata2019symmetry,borgnia2020non}.
While quantum algorithms for Hermitian spectral estimation are extensively studied~\cite{kitaev2002classical,Somma2019,OBrien2019,lin2022heisenberg,wan2022randomized,ni2023low,ding2023even,ding2024quantum,alase2024resolvent,castaldo2025heisenberg}, corresponding algorithms for non-Hermitian problems have appeared only recently~\cite{shao2022computing,zhang2024exponential,lin2025quantum,zhang2025heisenberg,low2026quantum}. 
For non-normal operators, eigenvalues can be highly sensitive to perturbations~\cite{trefethen2005spectra}, so in the presence of noise and modeling error it can be difficult to distinguish the exact spectrum from nearby points with comparably large response~\cite{okuma2021non}.

Are there fundamental obstacles to computing non-Hermitian spectra on quantum computers? In this Letter, we first show that for a local non-Hermitian operator $A$, deciding whether a point $z\in\CC$ is $\epsilon$-close to the spectrum is $\PSPACE$-hard in the worst case. We therefore turn to the \emph{pseudospectrum}~\cite{trefethen2005spectra}, which is defined using the resolvent $(A-zI)^{-1}$ and records where small perturbations can induce large spectral shifts. Since $z\in\Lambda_\epsilon(A)$ if and only if the smallest singular value of $A-zI$ is at most $\epsilon$, pseudospectrum membership reduces to singular-value estimation. We prove that this problem is $\QMA$-complete for inverse-polynomially small $\epsilon$, which places it in the same complexity class as estimating the ground state energy of local Hamiltonian problems~\cite{kitaev2002classical}. 
While QMA-complete problems remain intractable in the worst case (unless $\mathrm{BQP}=\mathrm{QMA}$), this characterization shows that pseudospectrum membership is a natural quantum verification problem, analogous to the role of the local Hamiltonian problem in Hermitian many-body physics.

\begin{figure}
    \centering
    \includegraphics[width=0.45\textwidth]{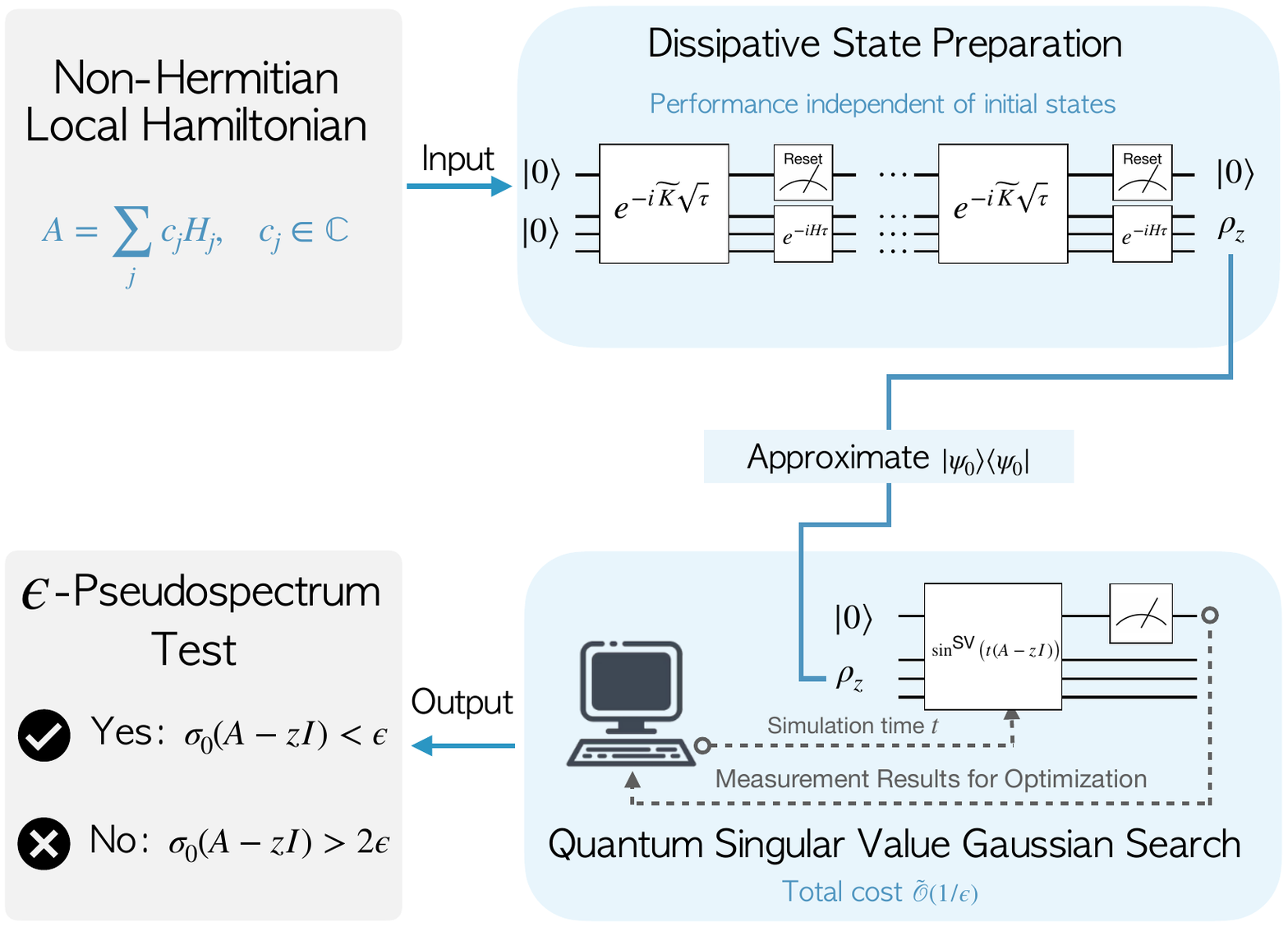}
    \caption{A two-pronged approach for estimating non-Hermitian pseudospectra on a quantum computer. }
    \label{fig:schematic}
\end{figure}

Can one design a concrete quantum framework for pseudospectrum estimation? We develop a two-pronged approach toward this goal (\cref{fig:schematic}). The first component is Quantum Singular-value Gaussian-filtered Search (QSIGS), which combines quantum singular value transformation (QSVT)~\cite{gilyen2019quantum} with classical post-processing to estimate the smallest singular value with Heisenberg-limited query scaling. The second component is an algorithmic Lindbladian protocol that dissipatively prepares approximate ground right singular vectors, addressing the need for an initial state with non-negligible overlap with the target singular vector~\cite{low2026quantum,zhang2024exponential,zhang2025heisenberg}. For this preparation step, we prove convergence for a continuous-time Lindbladian in the Hatano--Nelson model at $z=0$, and we give structural justification together with model-specific numerical evidence in broader settings.
In general, this dissipative state-preparation protocol should be understood as a heuristic, instance-dependent procedure rather than a universal efficient algorithm for arbitrary non-Hermitian many-body Hamiltonians.

Finally, can such protocols be implemented on current gate-based quantum hardware? As a proof of principle demonstration, we implement the full pipeline on a trapped-ion quantum computer to probe the pseudospectrum near the exceptional point of the qubit Hamiltonian $H(g)=X-igZ$~\cite{el2018non}. The experiment correctly distinguishes points inside and outside the target $\epsilon$-pseudospectrum, showing that the combined procedure already captures the relevant spectral instability in a minimal setting.

\section{Spectra and pseudospectra}
For $\epsilon>0$, the $\epsilon$-pseudospectrum of an operator $A$ is defined by $\Lambda_\epsilon(A):=\{z\in\CC\colon \|(A-zI)^{-1}\|\ge \epsilon^{-1}\}$. 
Let $A_z=A-zI$. If $z$ is an eigenvalue of $A$, then $z\in\Lambda_\epsilon(A)$ for every $\epsilon>0$. The pseudospectrum therefore enlarges the spectrum by including not only exact eigenvalues but also nearby points where the resolvent $A_z^{-1}$ is large. 

For our purposes, the key observation is that $z\in\Lambda_\epsilon(A)$ if and only if $\sigma_0(A_z)\le \epsilon$, where $\sigma_0(\cdot)$ denotes the smallest singular value. Equivalently, the promise problem of deciding 
\begin{equation}\label{eqn:decision}
\mbox{whether}\quad \sigma_0(A_z)\le \epsilon \quad \mbox{or} \quad \sigma_0(A_z)\ge 2\epsilon 
\end{equation}
is a singular-value estimation problem at precision $\epsilon$. This is the computational formulation used throughout the Letter. 

\section{Complexity-theoretic separation}

We work with local non-Hermitian operators built from few-body terms. Given such an operator $A$ and a point $z\in\CC$, one can ask two closely related but distinct questions: does $z$ lie in the $\epsilon$-pseudospectrum, or is $z$ within distance $\epsilon$ of an actual eigenvalue? For inverse-polynomially small $\epsilon$, we show that pseudospectrum membership is $\QMA$-complete for $4$-local operators, whereas deciding whether $z$ is $\epsilon$-close to the spectrum is $\PSPACE$-hard already for $5$-local operators.

The $\QMA$-completeness of pseudospectrum membership is conceptually natural. By the singular-value characterization above, $\sigma_0(A_z)^2$ is the ground energy of the Hermitian operator $H_z=A_z^\dagger A_z$. Thus, deciding whether $z\in\Lambda_\epsilon(A)$ is a local-Hamiltonian-type problem. The spectrum problem is harder for a specifically non-Hermitian reason. Strongly non-normal matrices can amplify an exponentially small energy scale into an order-one spectral displacement. It is known that estimating spectra of a local Hamiltonian to exponential accuracy is $\PSPACE$-complete, and deciding spectral proximity is therefore $\PSPACE$-hard. See~\cref{append:hardness-non-hermitian-hamiltonian} for the formal statements and proofs.
This separation identifies pseudospectrum membership, rather than exact spectral proximity, as a natural target for quantum algorithms on non-Hermitian many-body systems.

\section{Quantum singular-value Gaussian-filtered search (QSIGS)}

For Hermitian problems, starting from an initial state $\rho_0$ (pure or mixed), various phase-estimation methods~\cite{lin2022heisenberg,wan2022randomized,ni2023low,ding2023even,ding2024quantum} extract spectral information from the correlation function $Z(t)=\Tr[\rho_0 e^{-iHt}]$ and recover eigenvalues from its Fourier transform. For instance, the QMEGS algorithm~\cite{ding2024quantum} achieves Heisenberg-limited scaling by sampling $t$ from a Gaussian distribution and moving the spectral reconstruction to classical post-processing.

We now adapt this viewpoint to the singular values of a non-Hermitian operator. Consider the task of estimating the smallest singular value of $A_z$. A direct reduction to the Hermitian matrix $H_z$ is not satisfactory: its eigenvalues are $\sigma_j^2$, so estimating a small singular value $\sigma_0$ to additive accuracy $\epsilon$ requires resolving $\sigma_0^2$ to accuracy $\Or(\epsilon^2)$, which leads to query complexity $\Or(\epsilon^{-2})$. One can avoid this squaring loss by passing to the Hermitian dilation $\begin{psmallmatrix}0 & A_z^{\dagger} \\ A_z & 0\end{psmallmatrix}$, whose eigenvalues are $\pm \sigma_j$. This restores Heisenberg scaling, but for the block-encoding model it requires controlled access to the encoding unitary.

QSIGS instead works directly with a block-encoding of $A_z$. Using QSVT, we implement the singular-value transformation $\sinsv(t A_z)$, whose action depends on each singular value through $\sin(t\sigma_j)$. Measuring the ancilla register then produces a bounded random variable whose expectation is a linear combination of $\cos(2\sigma_j t)$, with coefficients determined by the overlaps of the input state with the right singular vectors. Sampling $t$ from a Gaussian distribution and forming the corresponding filter function therefore converts singular-value estimation into a one-dimensional classical peak-finding problem: the filter develops Gaussian peaks near the singular values. If the target state is $\ket{\psi}$ and the input is $\rho_0$ with overlap $p_0=\Tr[\rho_0 \ketbra{\psi}{\psi}]$, then when $p_0$ is sufficiently large, the maximizer of the filter identifies the desired singular value, and in this case $\sigma_0$. In particular, QSIGS gives a Heisenberg-limited singular-value estimation primitive in the dominant-overlap regime relevant to our pipeline. This achieves the Heisenberg-limited scaling, while avoiding controlled use of the block-encoding. \cref{append:qsigs} gives the algorithm and formal guarantees.

\section{Dissipative state preparation}\label{sec:NHEP-DD} 

QSIGS, like other filtering-based spectral algorithms, gives an efficient decision primitive for pseudospectrum membership only when the input state has non-negligible overlap with the target singular vector. Specifically, achieving constant success probability generally requires $\Omega(1/p_0)$ repetitions, or $\Omega\left(1/\sqrt{p_0}\right)$ with amplitude amplification. We therefore seek a preparation method that yields appreciable overlap with the ground right singular vector of $A_z$ across many values of $z$.

We use dissipative state preparation for this purpose~\cite{verstraete2009quantum,RoyChalkerGornyiEtAl2020,ChenKastoryanoBrandaoEtAl2025,DingLiLin2025,ding2024single,LiZhanLin2025,Lloyd2025Quasiparticle,ZhanDingHuhnEtAl2026}. Rather than variationally tuning an ansatz for each shifted operator $A_z$~\cite{xie2024variational,zhang2025observation,zhao2023universal}, we construct an algorithmic Lindbladian whose fixed-point space contains the ground-state subspace of the positive semidefinite Hamiltonian $H_z$. Any pure state in that subspace is a right singular vector associated with the smallest singular value of $A_z$.

The construction starts from a collection of coupling operators $\{O_a\}$ and the Lindblad master equation
\begin{equation}\label{eqn:lindbladian}
    \dot{\rho} = \mathcal{L} [\rho] \coloneqq -i [H_z, \rho] + \sum_a K_a\rho K^\dag_a - \frac{1}{2}\{K^\dag_a K_a, \rho\},
\end{equation}
where the jump operators $\{K_a\}$ are defined by the weighted operator Fourier transform
\begin{equation}
    K_a = \int_{\mathbb{R}} f(s) e^{i H_z s} O_a e^{-i H_z s} \mathrm{d}s.\label{eqn:energy-basis-jump}
\end{equation}
Here $f$ is chosen so that its Fourier transform satisfies $\hat{f}(\omega)=0$ for $\omega \ge 0$. In the energy eigenbasis of $H_z$, this means that the jumps only permit energy-lowering transitions and suppress energy-increasing ones. As a result, the ground state of $H_z$ is a stationary state of the Lindbladian. 

In practice, we implement not the continuous-time flow itself but a discrete-time completely positive map obtained by composing Hamiltonian evolution under $H_z$ with a Stinespring realization of the dissipative step. 
The coupling operators $\{O_a\}$ should be chosen so that their filtered jumps $\{K_a\}$ produce energy-lowering transitions that connect excited states to the ground-energy manifold of $H_z$.
Because the jump operators annihilate the ground-state subspace of $H_z$, this map preserves the same fixed points for any step size $\tau$. This keeps the circuit primitive shallow. \cref{append:circuit-implementation} describes the discrete-time implementation of this dynamics. After simulating for some time $t_{\rm sim}$, we obtain a state $\rho_z=\rho(t_{\rm sim})$ with large overlap with the target singular vector, which is then used as the initial state for QSIGS.

\section{Application to the Hatano--Nelson model}

The efficiency of dissipative preparation is governed by its mixing time, namely the time required for the Lindbladian dynamics to concentrate on the target ground-state subspace. This depends on both the shifted Hamiltonian $H_z$ and the chosen couplings $\{O_a\}$, and is generally difficult to analyze. The Hatano--Nelson model~\cite{hatano1996localization,hatano1998non} provides a useful test case because its non-Hermitian structure is simple enough to admit both rigorous analysis and informative numerics.

The model describes a one-dimensional lattice with asymmetric nearest-neighbor hopping,
\begin{equation}
    H_{\text{HN}} = \sum_{j} (J + \gamma) \ketbra{j+1}{j} + (J - \gamma)\ketbra{j}{j+1},
\end{equation}
where $J\pm \gamma$ are the hopping amplitudes and the non-Hermiticity arises from the non-reciprocity $\gamma\neq 0$. It is a canonical example of the non-Hermitian skin effect, and its spectrum is highly sensitive to boundary conditions.
Under open boundary conditions (OBC), the eigenvalues of $H_{\rm HN}$ are given by $E^{(j)}_{\rm OBC} = 2\sqrt{J^2-\gamma^2}\cos(\frac{j\pi}{n+1})$, $j=1,2,\dots,n$, where $n$ is the total number of sites.
Under periodic boundary conditions (PBC), the eigenvalues lie on an ellipse that encloses the OBC spectrum in the complex plane~\cite{feng2025numerical}: $E^{(j)}_{\rm PBC} = 2J\cos(2j\pi/n) -2i\gamma \sin(2j\pi/n)$, $j=1,2,\dots, n$.
We illustrate both spectra in~\cref{fig:HN}(a).

We focus on preparing the ground right singular vector of $A_z = H^{\rm PBC}_{\rm HN}-zI$. For $z=0$, the associated Hermitian problem $H_0 = H^\dagger_{\rm HN} H_{\rm HN}$ is analytically tractable. A convenient choice is the pair of coupling operators
\begin{align}\label{eqn:hn-coupling}
    O_0 = \sum^{n-1}_{k=0} e^{2\pi i k/n} \ketbra{k}{k},\quad O_1 = \sum^{n-1}_{k=0} e^{-2\pi i k/n} \ketbra{k}{k}.
\end{align}
These operators act as one-step moves in opposite directions in the Fourier basis, so at $z=0$ the dissipative dynamics reduces to an effectively one-dimensional population process with an absorbing ground-state sector. As shown in~\cref{append:hatano-nelson-mixing}, the mixing time satisfies $t_{\rm mix}(\epsilon)\le n/4+\Or(\sqrt{n\log(1/\epsilon)})$, which in particular scales linearly in the system size $n$.

For general $z$, the shifted Hamiltonian $H_z = (H_{\rm HN}-zI)^\dagger(H_{\rm HN}-zI)$ is no longer analytically diagonalizable, so we turn to numerical evidence. In addition to the couplings in~\cref{eqn:hn-coupling}, we include a reflection-type coupling that improves mixing across the spectrum of $H_z$. The heat map in~\cref{fig:HN}(b) records the simulation time needed, starting from the highest-energy state, to reach the threshold $|\Tr(H_z\rho(t))-E_0(H_z)|\le 10^{-3}$ for $n=20$, $J=1$, $\gamma=0.8$, using the discrete-time dynamics with step size $\tau=0.1$. 
Over the plotted region, which contains the full PBC spectrum, the required simulation time stays below $30$. The lighter region at large $|\operatorname{Im} z|$ should be read as an instance-dependent slowdown, since the same fixed couplings still reach the stated energy threshold at all sampled points, but this does not imply a uniform efficient bound away from the analytically tractable point $z=0$.
Similar convergence behaviors are observed in the Hatano--Nelson model under OBC;
see~\cref{append:hatano-nelson-mixing} for details.

\begin{figure}
    \includegraphics[width=0.5\textwidth]{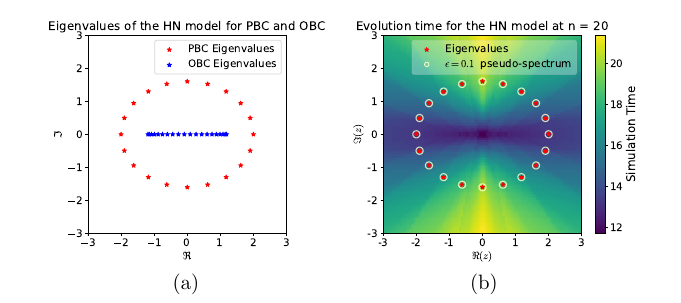}
    \caption{Preparing ground singular vectors of the Hatano-Nelson model using dissipative dynamics.
    (a) The eigenvalues of the Hatano-Nelson model under OBC and PBC ($n=20$).
    (b) The $0.1$-pseudospectrum of the Hatano-Nelson model under PBC, and the mixing time of the Lindbladian dynamics for preparing the ground singular vector of $H_{\rm HN}-zI$.
    }
    \label{fig:HN}
    \vspace{-4mm}
\end{figure}

\begin{figure*}
    \includegraphics[width=0.99\textwidth]{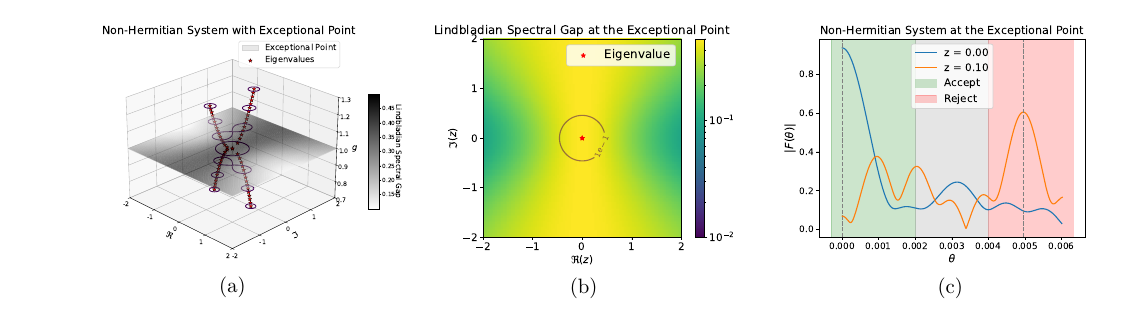}
    \caption{Probing exceptional points of a non-Hermitian qubit system using IonQ quantum computer. 
    (a) Illustration of eigenvalues and pseudospectra of the model Hamiltonian $H(g)$. The closed curves at each level shows the $0.1$-pseudospectrum at the corresponding $g$ value.
    (b) The $0.1$-pseudospectrum of $H(1)$, and the spectral gap of $\mathcal{L}_z$ over the complex plane.
    (c) The filter functions in QSIGS constructed using the measurement data from IonQ. The peak of the filter function indicates $\sigma_0(H(1)-zI)$. We fix $\epsilon=0.002$, and the results show that $z=0$ lies in the $\epsilon$-pseudospectrum while $z=0.1$ does not.}
    \label{fig:NH}
\end{figure*}

\section{Probing exceptional points on an IonQ device}\label{sec:exceptional-demo}

We now implement the complete pipeline, including dissipative preparation and QSIGS, on trapped-ion hardware. Our goal is to probe exceptional-point pseudospectra, where non-Hermitian spectral instability is most pronounced. At an order-$n$ exceptional point ($n \ge 2$), $n$ eigenvalues and eigenvectors coalesce, so the matrix becomes defective. In this regime, the pseudospectrum records the instability directly: a perturbation of size $\epsilon$ can produce an eigenvalue displacement of order $\epsilon^{1/n}$.

As a minimal test case, we use the qubit Hamiltonian $H(g)=X-igZ$~\cite{el2018non}. Its eigenvalues are real for $g<1$ and imaginary for $g>1$, with the transition at $g=1$ marking a second-order exceptional point; see \cref{fig:NH}(a). At this point, the $\epsilon$-pseudospectrum is a disc centered at the origin with radius $r=\sqrt{2\epsilon+\epsilon^2}$; see~\cref{append:numerical}. This gives a clean benchmark for pseudospectrum membership: points near the origin should be accepted, while points outside this radius should be rejected.

We perform this test on IonQ Forte. For each target point $z$, we first use the dissipative protocol to prepare an approximate ground right singular vector of $A_z=H-zI$, and then apply QSIGS to estimate $\sigma_0(A_z)$. In the hardware implementation, the Lindbladian uses a single Pauli-$X$ coupling operator. Because IonQ does not support mid-circuit reset, we defer the reset operations by introducing additional ancillas and measuring them only at the end of the circuit. We then implement the QSIGS subroutine through synthesized block-encodings of the required singular-value transformation.

The resulting filter functions are shown in~\cref{fig:NH}(c). For each value of $z$, we use $50000$ measurement shots and extract $\sigma_0(A_z)$ from the location of the filter peak. At $z=0$, the peak lies well below the threshold $\epsilon=0.002$, so the point is correctly identified as lying in the $\epsilon$-pseudospectrum. At $z=0.1$, the peak lies beyond $2\epsilon$, so the point is correctly rejected. These outcomes agree with the exact pseudospectral radius $r\approx 0.063$ at $g=1$. 
We note that this experiment is intended to illustrate the algorithmic pipeline on gate-based quantum computers, rather than a demonstration of quantum advantage.
The circuit construction, experimental parameters and additional implementation details are described in~\cref{append:numerical}.

\section{Discussion}

This work identifies pseudospectrum membership as a natural computational target for non-Hermitian many-body physics and gives a concrete quantum framework for it. QSIGS achieves Heisenberg-limited query scaling for singular-value estimation. The remaining challenge is the dissipative preparation step, whose complexity depends on both the mixing time and the cost of implementing the jump operators. In our current construction, the dissipative dynamics is generated by $H_z=A_z^\dagger A_z$, so realizing the jump operators requires access to $H_z$ and may depend sensitively on the local singular-value structure near the spectrum. Our analysis establishes rigorous convergence for the Hatano--Nelson model, but extending such guarantees to broader classes of systems remains open. Free-fermion systems provide a natural next test bed; however, because $H_z$ generally turns quadratic fermionic problems into quartic ones, new techniques will be needed to analyze dissipative mixing in non-Hermitian many-body settings.

The trapped-ion experiment should be viewed as a proof of principle. The next steps are to reduce the overhead of block-encoding and dissipative primitives, and extend the method to larger local models and broader regions of the complex plane. Progress on these fronts would help make pseudospectrum estimation a practical primitive for quantum simulation of non-Hermitian systems.

\section{Acknowledgment}
This work is partially supported by the Simons Quantum Postdoctoral Fellowship (J.L.), by the U.S. Department of Energy, Office of Science, Accelerated Research in Quantum Computing Centers, Quantum Utility through Advanced Computational Quantum Algorithms, grant no. DE-SC0025572 (L.L.) and grant no. DE-SC0025341 (G.Y., X.W.), the U.S. National Science Foundation grant CCF-1942837 (CAREER) (G.Y., X.W.), a Sloan research fellowship (X.W.), and a Simons Investigator Award in Mathematics through Grant No. 825053 (J.L., L.L.). 
We are also grateful to the access to IonQ
machines provided by the National Quantum Laboratory
(QLab) at UMD.

\bibliographystyle{apsrev4-2}
\bibliography{ref}

\begin{thebibliography}{58}%
\makeatletter
\providecommand \@ifxundefined [1]{%
 \@ifx{#1\undefined}
}%
\providecommand \@ifnum [1]{%
 \ifnum #1\expandafter \@firstoftwo
 \else \expandafter \@secondoftwo
 \fi
}%
\providecommand \@ifx [1]{%
 \ifx #1\expandafter \@firstoftwo
 \else \expandafter \@secondoftwo
 \fi
}%
\providecommand \natexlab [1]{#1}%
\providecommand \enquote  [1]{``#1''}%
\providecommand \bibnamefont  [1]{#1}%
\providecommand \bibfnamefont [1]{#1}%
\providecommand \citenamefont [1]{#1}%
\providecommand \href@noop [0]{\@secondoftwo}%
\providecommand \href [0]{\begingroup \@sanitize@url \@href}%
\providecommand \@href[1]{\@@startlink{#1}\@@href}%
\providecommand \@@href[1]{\endgroup#1\@@endlink}%
\providecommand \@sanitize@url [0]{\catcode `\\12\catcode `\$12\catcode `\&12\catcode `\#12\catcode `\^12\catcode `\_12\catcode `\%12\relax}%
\providecommand \@@startlink[1]{}%
\providecommand \@@endlink[0]{}%
\providecommand \url  [0]{\begingroup\@sanitize@url \@url }%
\providecommand \@url [1]{\endgroup\@href {#1}{\urlprefix }}%
\providecommand \urlprefix  [0]{URL }%
\providecommand \Eprint [0]{\href }%
\providecommand \doibase [0]{https://doi.org/}%
\providecommand \selectlanguage [0]{\@gobble}%
\providecommand \bibinfo  [0]{\@secondoftwo}%
\providecommand \bibfield  [0]{\@secondoftwo}%
\providecommand \translation [1]{[#1]}%
\providecommand \BibitemOpen [0]{}%
\providecommand \bibitemStop [0]{}%
\providecommand \bibitemNoStop [0]{.\EOS\space}%
\providecommand \EOS [0]{\spacefactor3000\relax}%
\providecommand \BibitemShut  [1]{\csname bibitem#1\endcsname}%
\let\auto@bib@innerbib\@empty
\bibitem [{\citenamefont {Ashida}\ \emph {et~al.}(2020)\citenamefont {Ashida}, \citenamefont {Gong},\ and\ \citenamefont {Ueda}}]{ashida2020non}%
  \BibitemOpen
  \bibfield  {author} {\bibinfo {author} {\bibfnamefont {Y.}~\bibnamefont {Ashida}}, \bibinfo {author} {\bibfnamefont {Z.}~\bibnamefont {Gong}},\ and\ \bibinfo {author} {\bibfnamefont {M.}~\bibnamefont {Ueda}},\ }\href {https://doi.org/https://doi.org/10.1080/00018732.2021.1876991} {\bibfield  {journal} {\bibinfo  {journal} {Adv. Phys}\ }\textbf {\bibinfo {volume} {69}},\ \bibinfo {pages} {249} (\bibinfo {year} {2020})}\BibitemShut {NoStop}%
\bibitem [{\citenamefont {Xiao}\ \emph {et~al.}(2024)\citenamefont {Xiao}, \citenamefont {Chu}, \citenamefont {Lin}, \citenamefont {Lin}, \citenamefont {Yi}, \citenamefont {Cai},\ and\ \citenamefont {Xue}}]{xiao2024non}%
  \BibitemOpen
  \bibfield  {author} {\bibinfo {author} {\bibfnamefont {L.}~\bibnamefont {Xiao}}, \bibinfo {author} {\bibfnamefont {Y.}~\bibnamefont {Chu}}, \bibinfo {author} {\bibfnamefont {Q.}~\bibnamefont {Lin}}, \bibinfo {author} {\bibfnamefont {H.}~\bibnamefont {Lin}}, \bibinfo {author} {\bibfnamefont {W.}~\bibnamefont {Yi}}, \bibinfo {author} {\bibfnamefont {J.}~\bibnamefont {Cai}},\ and\ \bibinfo {author} {\bibfnamefont {P.}~\bibnamefont {Xue}},\ }\href {https://doi.org/10.1103/PhysRevLett.133.180801} {\bibfield  {journal} {\bibinfo  {journal} {Phys. Rev. Lett.}\ }\textbf {\bibinfo {volume} {133}},\ \bibinfo {pages} {180801} (\bibinfo {year} {2024})}\BibitemShut {NoStop}%
\bibitem [{\citenamefont {Wiersig}(2020)}]{wiersig2020review}%
  \BibitemOpen
  \bibfield  {author} {\bibinfo {author} {\bibfnamefont {J.}~\bibnamefont {Wiersig}},\ }\href {https://doi.org/10.1364/PRJ.396115} {\bibfield  {journal} {\bibinfo  {journal} {Photon. Res.}\ }\textbf {\bibinfo {volume} {8}},\ \bibinfo {pages} {1457} (\bibinfo {year} {2020})}\BibitemShut {NoStop}%
\bibitem [{\citenamefont {Chaduteau}\ \emph {et~al.}(2026)\citenamefont {Chaduteau}, \citenamefont {Lee},\ and\ \citenamefont {Schindler}}]{chaduteau2026lindbladian}%
  \BibitemOpen
  \bibfield  {author} {\bibinfo {author} {\bibfnamefont {A.}~\bibnamefont {Chaduteau}}, \bibinfo {author} {\bibfnamefont {D.~K.}\ \bibnamefont {Lee}},\ and\ \bibinfo {author} {\bibfnamefont {F.}~\bibnamefont {Schindler}},\ }\href {https://doi.org/10.1103/tpfc-n3bq} {\bibfield  {journal} {\bibinfo  {journal} {Phys. Rev. Lett.}\ }\textbf {\bibinfo {volume} {136}},\ \bibinfo {pages} {113401} (\bibinfo {year} {2026})}\BibitemShut {NoStop}%
\bibitem [{\citenamefont {Almanakly}\ \emph {et~al.}(2026)\citenamefont {Almanakly}, \citenamefont {Assouly}, \citenamefont {Kang}, \citenamefont {Gingras}, \citenamefont {Niedzielski}, \citenamefont {Stickler}, \citenamefont {Schwartz}, \citenamefont {Serniak}, \citenamefont {Hays}, \citenamefont {Grover} \emph {et~al.}}]{almanakly2026probing}%
  \BibitemOpen
  \bibfield  {author} {\bibinfo {author} {\bibfnamefont {A.}~\bibnamefont {Almanakly}}, \bibinfo {author} {\bibfnamefont {R.}~\bibnamefont {Assouly}}, \bibinfo {author} {\bibfnamefont {H.~H.}\ \bibnamefont {Kang}}, \bibinfo {author} {\bibfnamefont {M.}~\bibnamefont {Gingras}}, \bibinfo {author} {\bibfnamefont {B.~M.}\ \bibnamefont {Niedzielski}}, \bibinfo {author} {\bibfnamefont {H.}~\bibnamefont {Stickler}}, \bibinfo {author} {\bibfnamefont {M.~E.}\ \bibnamefont {Schwartz}}, \bibinfo {author} {\bibfnamefont {K.}~\bibnamefont {Serniak}}, \bibinfo {author} {\bibfnamefont {M.}~\bibnamefont {Hays}}, \bibinfo {author} {\bibfnamefont {J.~A.}\ \bibnamefont {Grover}}, \emph {et~al.},\ }\href {https://doi.org/10.1103/gd4s-fgsswt} {\bibfield  {journal} {\bibinfo  {journal} {Phys. Rev. Lett.}\ }\textbf {\bibinfo {volume} {136}},\ \bibinfo {pages} {183601} (\bibinfo {year} {2026})}\BibitemShut {NoStop}%
\bibitem [{\citenamefont {Tao}\ \emph {et~al.}(2026)\citenamefont {Tao}, \citenamefont {Mercado-Gutierrez}, \citenamefont {Zhao},\ and\ \citenamefont {Spielman}}]{tao2025imaginary}%
  \BibitemOpen
  \bibfield  {author} {\bibinfo {author} {\bibfnamefont {J.}~\bibnamefont {Tao}}, \bibinfo {author} {\bibfnamefont {E.}~\bibnamefont {Mercado-Gutierrez}}, \bibinfo {author} {\bibfnamefont {M.}~\bibnamefont {Zhao}},\ and\ \bibinfo {author} {\bibfnamefont {I.}~\bibnamefont {Spielman}},\ }\href {https://doi.org/10.1103/tpfc-n3bq} {\bibfield  {journal} {\bibinfo  {journal} {Phys. Rev. Lett.}\ }\textbf {\bibinfo {volume} {136}},\ \bibinfo {pages} {016603} (\bibinfo {year} {2026})}\BibitemShut {NoStop}%
\bibitem [{\citenamefont {Yao}\ and\ \citenamefont {Wang}(2018)}]{yao2018edge}%
  \BibitemOpen
  \bibfield  {author} {\bibinfo {author} {\bibfnamefont {S.}~\bibnamefont {Yao}}\ and\ \bibinfo {author} {\bibfnamefont {Z.}~\bibnamefont {Wang}},\ }\href {https://doi.org/https://doi.org/10.1103/PhysRevLett.121.086803} {\bibfield  {journal} {\bibinfo  {journal} {Phys. Rev. Lett.}\ }\textbf {\bibinfo {volume} {121}},\ \bibinfo {pages} {086803} (\bibinfo {year} {2018})}\BibitemShut {NoStop}%
\bibitem [{\citenamefont {D{\'o}ra}\ and\ \citenamefont {Moca}(2020)}]{dora2020quantum}%
  \BibitemOpen
  \bibfield  {author} {\bibinfo {author} {\bibfnamefont {B.}~\bibnamefont {D{\'o}ra}}\ and\ \bibinfo {author} {\bibfnamefont {C.~P.}\ \bibnamefont {Moca}},\ }\href {https://doi.org/https://doi.org/10.1103/PhysRevLett.124.136802} {\bibfield  {journal} {\bibinfo  {journal} {Phys. Rev. Lett.}\ }\textbf {\bibinfo {volume} {124}},\ \bibinfo {pages} {136802} (\bibinfo {year} {2020})}\BibitemShut {NoStop}%
\bibitem [{\citenamefont {Gopalakrishnan}\ and\ \citenamefont {Gullans}(2021)}]{gopalakrishnan2021entanglement}%
  \BibitemOpen
  \bibfield  {author} {\bibinfo {author} {\bibfnamefont {S.}~\bibnamefont {Gopalakrishnan}}\ and\ \bibinfo {author} {\bibfnamefont {M.~J.}\ \bibnamefont {Gullans}},\ }\href {https://doi.org/https://doi.org/10.1103/PhysRevLett.126.170503} {\bibfield  {journal} {\bibinfo  {journal} {Phys. Rev. Lett.}\ }\textbf {\bibinfo {volume} {126}},\ \bibinfo {pages} {170503} (\bibinfo {year} {2021})}\BibitemShut {NoStop}%
\bibitem [{\citenamefont {Zhang}\ \emph {et~al.}(2025{\natexlab{a}})\citenamefont {Zhang}, \citenamefont {Carrasquilla},\ and\ \citenamefont {Kim}}]{zhang2025observation}%
  \BibitemOpen
  \bibfield  {author} {\bibinfo {author} {\bibfnamefont {Y.}~\bibnamefont {Zhang}}, \bibinfo {author} {\bibfnamefont {J.}~\bibnamefont {Carrasquilla}},\ and\ \bibinfo {author} {\bibfnamefont {Y.~B.}\ \bibnamefont {Kim}},\ }\href {https://doi.org/https://doi.org/10.1038/s41467-025-57930-3} {\bibfield  {journal} {\bibinfo  {journal} {Nat. Commun.}\ }\textbf {\bibinfo {volume} {16}},\ \bibinfo {pages} {3286} (\bibinfo {year} {2025}{\natexlab{a}})}\BibitemShut {NoStop}%
\bibitem [{\citenamefont {Bender}\ and\ \citenamefont {Boettcher}(1998)}]{bender1998real}%
  \BibitemOpen
  \bibfield  {author} {\bibinfo {author} {\bibfnamefont {C.~M.}\ \bibnamefont {Bender}}\ and\ \bibinfo {author} {\bibfnamefont {S.}~\bibnamefont {Boettcher}},\ }\href {https://doi.org/https://doi.org/10.1103/PhysRevLett.80.5243} {\bibfield  {journal} {\bibinfo  {journal} {Phys. Rev. Lett.}\ }\textbf {\bibinfo {volume} {80}},\ \bibinfo {pages} {5243} (\bibinfo {year} {1998})}\BibitemShut {NoStop}%
\bibitem [{\citenamefont {Bender}\ \emph {et~al.}(1999)\citenamefont {Bender}, \citenamefont {Boettcher},\ and\ \citenamefont {Meisinger}}]{bender1999PT}%
  \BibitemOpen
  \bibfield  {author} {\bibinfo {author} {\bibfnamefont {C.~M.}\ \bibnamefont {Bender}}, \bibinfo {author} {\bibfnamefont {S.}~\bibnamefont {Boettcher}},\ and\ \bibinfo {author} {\bibfnamefont {P.~N.}\ \bibnamefont {Meisinger}},\ }\href {https://doi.org/https://doi.org/10.1063/1.532860} {\bibfield  {journal} {\bibinfo  {journal} {J. Math. Phys.}\ }\textbf {\bibinfo {volume} {40}},\ \bibinfo {pages} {2201} (\bibinfo {year} {1999})}\BibitemShut {NoStop}%
\bibitem [{\citenamefont {El-Ganainy}\ \emph {et~al.}(2018)\citenamefont {El-Ganainy}, \citenamefont {Makris}, \citenamefont {Khajavikhan}, \citenamefont {Musslimani}, \citenamefont {Rotter},\ and\ \citenamefont {Christodoulides}}]{el2018non}%
  \BibitemOpen
  \bibfield  {author} {\bibinfo {author} {\bibfnamefont {R.}~\bibnamefont {El-Ganainy}}, \bibinfo {author} {\bibfnamefont {K.~G.}\ \bibnamefont {Makris}}, \bibinfo {author} {\bibfnamefont {M.}~\bibnamefont {Khajavikhan}}, \bibinfo {author} {\bibfnamefont {Z.~H.}\ \bibnamefont {Musslimani}}, \bibinfo {author} {\bibfnamefont {S.}~\bibnamefont {Rotter}},\ and\ \bibinfo {author} {\bibfnamefont {D.~N.}\ \bibnamefont {Christodoulides}},\ }\href {https://doi.org/https://doi.org/10.1038/nphys4323} {\bibfield  {journal} {\bibinfo  {journal} {Nat. Phys.}\ }\textbf {\bibinfo {volume} {14}},\ \bibinfo {pages} {11} (\bibinfo {year} {2018})}\BibitemShut {NoStop}%
\bibitem [{\citenamefont {Zhang}\ \emph {et~al.}(2022)\citenamefont {Zhang}, \citenamefont {Zhang}, \citenamefont {Lu},\ and\ \citenamefont {Chen}}]{zhang2022review}%
  \BibitemOpen
  \bibfield  {author} {\bibinfo {author} {\bibfnamefont {X.}~\bibnamefont {Zhang}}, \bibinfo {author} {\bibfnamefont {T.}~\bibnamefont {Zhang}}, \bibinfo {author} {\bibfnamefont {M.-H.}\ \bibnamefont {Lu}},\ and\ \bibinfo {author} {\bibfnamefont {Y.-F.}\ \bibnamefont {Chen}},\ }\href {https://doi.org/https://doi.org/10.1080/23746149.2022.2109431} {\bibfield  {journal} {\bibinfo  {journal} {Adv. Phys. X}\ }\textbf {\bibinfo {volume} {7}},\ \bibinfo {pages} {2109431} (\bibinfo {year} {2022})}\BibitemShut {NoStop}%
\bibitem [{\citenamefont {Heiss}(2012)}]{heiss2012physics}%
  \BibitemOpen
  \bibfield  {author} {\bibinfo {author} {\bibfnamefont {W.~D.}\ \bibnamefont {Heiss}},\ }\href {https://doi.org/https://doi.org/10.1088/1751-8113/45/44/444016} {\bibfield  {journal} {\bibinfo  {journal} {J. Phys. A: Math. Theor.}\ }\textbf {\bibinfo {volume} {45}},\ \bibinfo {pages} {444016} (\bibinfo {year} {2012})}\BibitemShut {NoStop}%
\bibitem [{\citenamefont {Kawabata}\ \emph {et~al.}(2019)\citenamefont {Kawabata}, \citenamefont {Shiozaki}, \citenamefont {Ueda},\ and\ \citenamefont {Sato}}]{kawabata2019symmetry}%
  \BibitemOpen
  \bibfield  {author} {\bibinfo {author} {\bibfnamefont {K.}~\bibnamefont {Kawabata}}, \bibinfo {author} {\bibfnamefont {K.}~\bibnamefont {Shiozaki}}, \bibinfo {author} {\bibfnamefont {M.}~\bibnamefont {Ueda}},\ and\ \bibinfo {author} {\bibfnamefont {M.}~\bibnamefont {Sato}},\ }\href {https://doi.org/https://doi.org/10.1103/PhysRevX.9.041015} {\bibfield  {journal} {\bibinfo  {journal} {Phys. Rev. X}\ }\textbf {\bibinfo {volume} {9}},\ \bibinfo {pages} {041015} (\bibinfo {year} {2019})}\BibitemShut {NoStop}%
\bibitem [{\citenamefont {Borgnia}\ \emph {et~al.}(2020)\citenamefont {Borgnia}, \citenamefont {Kruchkov},\ and\ \citenamefont {Slager}}]{borgnia2020non}%
  \BibitemOpen
  \bibfield  {author} {\bibinfo {author} {\bibfnamefont {D.~S.}\ \bibnamefont {Borgnia}}, \bibinfo {author} {\bibfnamefont {A.~J.}\ \bibnamefont {Kruchkov}},\ and\ \bibinfo {author} {\bibfnamefont {R.-J.}\ \bibnamefont {Slager}},\ }\href {https://doi.org/https://doi.org/10.1103/PhysRevLett.124.056802} {\bibfield  {journal} {\bibinfo  {journal} {Phys. Rev. Lett.}\ }\textbf {\bibinfo {volume} {124}},\ \bibinfo {pages} {056802} (\bibinfo {year} {2020})}\BibitemShut {NoStop}%
\bibitem [{\citenamefont {Kitaev}\ \emph {et~al.}(2002)\citenamefont {Kitaev}, \citenamefont {Shen},\ and\ \citenamefont {Vyalyi}}]{kitaev2002classical}%
  \BibitemOpen
  \bibfield  {author} {\bibinfo {author} {\bibfnamefont {A.~Y.}\ \bibnamefont {Kitaev}}, \bibinfo {author} {\bibfnamefont {A.}~\bibnamefont {Shen}},\ and\ \bibinfo {author} {\bibfnamefont {M.~N.}\ \bibnamefont {Vyalyi}},\ }\href@noop {} {\emph {\bibinfo {title} {Classical and Quantum Computation}}},\ \bibinfo {number} {47}\ (\bibinfo  {publisher} {American Mathematical Soc.},\ \bibinfo {year} {2002})\BibitemShut {NoStop}%
\bibitem [{\citenamefont {Somma}(2019)}]{Somma2019}%
  \BibitemOpen
  \bibfield  {author} {\bibinfo {author} {\bibfnamefont {R.~D.}\ \bibnamefont {Somma}},\ }\href {https://doi.org/10.1088/1367-2630/ab5c60} {\bibfield  {journal} {\bibinfo  {journal} {New J. Phys.}\ }\textbf {\bibinfo {volume} {21}},\ \bibinfo {pages} {123025} (\bibinfo {year} {2019})}\BibitemShut {NoStop}%
\bibitem [{\citenamefont {O'Brien}\ \emph {et~al.}(2019)\citenamefont {O'Brien}, \citenamefont {Tarasinski},\ and\ \citenamefont {Terhal}}]{OBrien2019}%
  \BibitemOpen
  \bibfield  {author} {\bibinfo {author} {\bibfnamefont {T.~E.}\ \bibnamefont {O'Brien}}, \bibinfo {author} {\bibfnamefont {B.}~\bibnamefont {Tarasinski}},\ and\ \bibinfo {author} {\bibfnamefont {B.~M.}\ \bibnamefont {Terhal}},\ }\href {https://doi.org/10.1088/1367-2630/aafb8e} {\bibfield  {journal} {\bibinfo  {journal} {New J. Phys.}\ }\textbf {\bibinfo {volume} {21}},\ \bibinfo {pages} {023022} (\bibinfo {year} {2019})}\BibitemShut {NoStop}%
\bibitem [{\citenamefont {Lin}\ and\ \citenamefont {Tong}(2022)}]{lin2022heisenberg}%
  \BibitemOpen
  \bibfield  {author} {\bibinfo {author} {\bibfnamefont {L.}~\bibnamefont {Lin}}\ and\ \bibinfo {author} {\bibfnamefont {Y.}~\bibnamefont {Tong}},\ }\href {https://doi.org/https://doi.org/10.1103/PRXQuantum.3.010318} {\bibfield  {journal} {\bibinfo  {journal} {PRX Quantum}\ }\textbf {\bibinfo {volume} {3}},\ \bibinfo {pages} {010318} (\bibinfo {year} {2022})}\BibitemShut {NoStop}%
\bibitem [{\citenamefont {Wan}\ \emph {et~al.}(2022)\citenamefont {Wan}, \citenamefont {Berta},\ and\ \citenamefont {Campbell}}]{wan2022randomized}%
  \BibitemOpen
  \bibfield  {author} {\bibinfo {author} {\bibfnamefont {K.}~\bibnamefont {Wan}}, \bibinfo {author} {\bibfnamefont {M.}~\bibnamefont {Berta}},\ and\ \bibinfo {author} {\bibfnamefont {E.~T.}\ \bibnamefont {Campbell}},\ }\href {https://doi.org/https://doi.org/10.1103/PhysRevLett.129.030503} {\bibfield  {journal} {\bibinfo  {journal} {Phys. Rev. Lett.}\ }\textbf {\bibinfo {volume} {129}},\ \bibinfo {pages} {030503} (\bibinfo {year} {2022})}\BibitemShut {NoStop}%
\bibitem [{\citenamefont {Ni}\ \emph {et~al.}(2023)\citenamefont {Ni}, \citenamefont {Li},\ and\ \citenamefont {Ying}}]{ni2023low}%
  \BibitemOpen
  \bibfield  {author} {\bibinfo {author} {\bibfnamefont {H.}~\bibnamefont {Ni}}, \bibinfo {author} {\bibfnamefont {H.}~\bibnamefont {Li}},\ and\ \bibinfo {author} {\bibfnamefont {L.}~\bibnamefont {Ying}},\ }\href {https://doi.org/https://doi.org/10.22331/q-2023-11-06-1165} {\bibfield  {journal} {\bibinfo  {journal} {Quantum}\ }\textbf {\bibinfo {volume} {7}},\ \bibinfo {pages} {1165} (\bibinfo {year} {2023})}\BibitemShut {NoStop}%
\bibitem [{\citenamefont {Ding}\ and\ \citenamefont {Lin}(2023)}]{ding2023even}%
  \BibitemOpen
  \bibfield  {author} {\bibinfo {author} {\bibfnamefont {Z.}~\bibnamefont {Ding}}\ and\ \bibinfo {author} {\bibfnamefont {L.}~\bibnamefont {Lin}},\ }\href {https://doi.org/https://doi.org/10.1103/PRXQuantum.4.020331} {\bibfield  {journal} {\bibinfo  {journal} {PRX Quantum}\ }\textbf {\bibinfo {volume} {4}},\ \bibinfo {pages} {020331} (\bibinfo {year} {2023})}\BibitemShut {NoStop}%
\bibitem [{\citenamefont {Ding}\ \emph {et~al.}(2024{\natexlab{a}})\citenamefont {Ding}, \citenamefont {Li}, \citenamefont {Lin}, \citenamefont {Ni}, \citenamefont {Ying},\ and\ \citenamefont {Zhang}}]{ding2024quantum}%
  \BibitemOpen
  \bibfield  {author} {\bibinfo {author} {\bibfnamefont {Z.}~\bibnamefont {Ding}}, \bibinfo {author} {\bibfnamefont {H.}~\bibnamefont {Li}}, \bibinfo {author} {\bibfnamefont {L.}~\bibnamefont {Lin}}, \bibinfo {author} {\bibfnamefont {H.}~\bibnamefont {Ni}}, \bibinfo {author} {\bibfnamefont {L.}~\bibnamefont {Ying}},\ and\ \bibinfo {author} {\bibfnamefont {R.}~\bibnamefont {Zhang}},\ }\href {https://doi.org/https://doi.org/10.22331/q-2024-10-02-1487} {\bibfield  {journal} {\bibinfo  {journal} {Quantum}\ }\textbf {\bibinfo {volume} {8}},\ \bibinfo {pages} {1487} (\bibinfo {year} {2024}{\natexlab{a}})}\BibitemShut {NoStop}%
\bibitem [{\citenamefont {Alase}\ and\ \citenamefont {Karuvade}(2024)}]{alase2024resolvent}%
  \BibitemOpen
  \bibfield  {author} {\bibinfo {author} {\bibfnamefont {A.}~\bibnamefont {Alase}}\ and\ \bibinfo {author} {\bibfnamefont {S.}~\bibnamefont {Karuvade}},\ }\href@noop {} {\bibfield  {journal} {\bibinfo  {journal} {arXiv preprint arXiv:2410.04837}\ } (\bibinfo {year} {2024})}\BibitemShut {NoStop}%
\bibitem [{\citenamefont {Castaldo}\ and\ \citenamefont {Corni}(2025)}]{castaldo2025heisenberg}%
  \BibitemOpen
  \bibfield  {author} {\bibinfo {author} {\bibfnamefont {D.}~\bibnamefont {Castaldo}}\ and\ \bibinfo {author} {\bibfnamefont {S.}~\bibnamefont {Corni}},\ }\href@noop {} {\bibfield  {journal} {\bibinfo  {journal} {arXiv preprint arXiv:2507.12438}\ } (\bibinfo {year} {2025})}\BibitemShut {NoStop}%
\bibitem [{\citenamefont {Shao}(2022)}]{shao2022computing}%
  \BibitemOpen
  \bibfield  {author} {\bibinfo {author} {\bibfnamefont {C.}~\bibnamefont {Shao}},\ }\bibfield  {journal} {\bibinfo  {journal} {ACM Trans. on Quantum Computing}\ }\textbf {\bibinfo {volume} {3}},\ \href {https://doi.org/10.1145/3527845} {10.1145/3527845} (\bibinfo {year} {2022})\BibitemShut {NoStop}%
\bibitem [{\citenamefont {Zhang}\ \emph {et~al.}(2024)\citenamefont {Zhang}, \citenamefont {Zhang}, \citenamefont {He},\ and\ \citenamefont {Yuan}}]{zhang2024exponential}%
  \BibitemOpen
  \bibfield  {author} {\bibinfo {author} {\bibfnamefont {X.-M.}\ \bibnamefont {Zhang}}, \bibinfo {author} {\bibfnamefont {Y.}~\bibnamefont {Zhang}}, \bibinfo {author} {\bibfnamefont {W.}~\bibnamefont {He}},\ and\ \bibinfo {author} {\bibfnamefont {X.}~\bibnamefont {Yuan}},\ }\href@noop {} {\bibfield  {journal} {\bibinfo  {journal} {arXiv preprint arXiv:2401.12091}\ } (\bibinfo {year} {2024})}\BibitemShut {NoStop}%
\bibitem [{\citenamefont {Lin}\ and\ \citenamefont {Shang}(2025)}]{lin2025quantum}%
  \BibitemOpen
  \bibfield  {author} {\bibinfo {author} {\bibfnamefont {H.}~\bibnamefont {Lin}}\ and\ \bibinfo {author} {\bibfnamefont {Y.}~\bibnamefont {Shang}},\ }\href@noop {} {\bibfield  {journal} {\bibinfo  {journal} {arXiv preprint arXiv:2502.18119}\ } (\bibinfo {year} {2025})}\BibitemShut {NoStop}%
\bibitem [{\citenamefont {Zhang}\ \emph {et~al.}(2025{\natexlab{b}})\citenamefont {Zhang}, \citenamefont {Wu},\ and\ \citenamefont {Yuan}}]{zhang2025heisenberg}%
  \BibitemOpen
  \bibfield  {author} {\bibinfo {author} {\bibfnamefont {Y.}~\bibnamefont {Zhang}}, \bibinfo {author} {\bibfnamefont {Y.}~\bibnamefont {Wu}},\ and\ \bibinfo {author} {\bibfnamefont {X.}~\bibnamefont {Yuan}},\ }\href@noop {} {\bibfield  {journal} {\bibinfo  {journal} {arXiv preprint arXiv:2510.19651}\ } (\bibinfo {year} {2025}{\natexlab{b}})}\BibitemShut {NoStop}%
\bibitem [{\citenamefont {Low}\ and\ \citenamefont {Su}(2026)}]{low2026quantum}%
  \BibitemOpen
  \bibfield  {author} {\bibinfo {author} {\bibfnamefont {G.~H.}\ \bibnamefont {Low}}\ and\ \bibinfo {author} {\bibfnamefont {Y.}~\bibnamefont {Su}},\ }\href {https://doi.org/https://doi.org/10.1137/24M1689363} {\bibfield  {journal} {\bibinfo  {journal} {SIAM Journal on Computing}\ }\textbf {\bibinfo {volume} {55}},\ \bibinfo {pages} {135} (\bibinfo {year} {2026})}\BibitemShut {NoStop}%
\bibitem [{\citenamefont {Trefethen}\ and\ \citenamefont {Embree}(2005)}]{trefethen2005spectra}%
  \BibitemOpen
  \bibfield  {author} {\bibinfo {author} {\bibfnamefont {L.~N.}\ \bibnamefont {Trefethen}}\ and\ \bibinfo {author} {\bibfnamefont {M.}~\bibnamefont {Embree}},\ }\href@noop {} {\emph {\bibinfo {title} {Spectra and Pseudospectra: The Behavior of Nonnormal Matrices and Operators}}}\ (\bibinfo  {publisher} {Princeton University Press},\ \bibinfo {year} {2005})\BibitemShut {NoStop}%
\bibitem [{\citenamefont {Okuma}\ and\ \citenamefont {Sato}(2021)}]{okuma2021non}%
  \BibitemOpen
  \bibfield  {author} {\bibinfo {author} {\bibfnamefont {N.}~\bibnamefont {Okuma}}\ and\ \bibinfo {author} {\bibfnamefont {M.}~\bibnamefont {Sato}},\ }\href {https://doi.org/https://doi.org/10.1103/PhysRevLett.126.176601} {\bibfield  {journal} {\bibinfo  {journal} {Phys. Rev. Lett.}\ }\textbf {\bibinfo {volume} {126}},\ \bibinfo {pages} {176601} (\bibinfo {year} {2021})}\BibitemShut {NoStop}%
\bibitem [{\citenamefont {Gily{\'e}n}\ \emph {et~al.}(2019)\citenamefont {Gily{\'e}n}, \citenamefont {Su}, \citenamefont {Low},\ and\ \citenamefont {Wiebe}}]{gilyen2019quantum}%
  \BibitemOpen
  \bibfield  {author} {\bibinfo {author} {\bibfnamefont {A.}~\bibnamefont {Gily{\'e}n}}, \bibinfo {author} {\bibfnamefont {Y.}~\bibnamefont {Su}}, \bibinfo {author} {\bibfnamefont {G.~H.}\ \bibnamefont {Low}},\ and\ \bibinfo {author} {\bibfnamefont {N.}~\bibnamefont {Wiebe}},\ }in\ \href {https://doi.org/https://doi.org/10.1145/3313276.3316366} {\emph {\bibinfo {booktitle} {Proceedings of the 51st Annual ACM SIGACT Symposium on Theory of Computing}}}\ (\bibinfo {year} {2019})\ pp.\ \bibinfo {pages} {193--204}\BibitemShut {NoStop}%
\bibitem [{\citenamefont {Verstraete}\ \emph {et~al.}(2009)\citenamefont {Verstraete}, \citenamefont {Wolf},\ and\ \citenamefont {Ignacio~Cirac}}]{verstraete2009quantum}%
  \BibitemOpen
  \bibfield  {author} {\bibinfo {author} {\bibfnamefont {F.}~\bibnamefont {Verstraete}}, \bibinfo {author} {\bibfnamefont {M.~M.}\ \bibnamefont {Wolf}},\ and\ \bibinfo {author} {\bibfnamefont {J.}~\bibnamefont {Ignacio~Cirac}},\ }\href {https://doi.org/https://doi.org/10.1038/nphys1342} {\bibfield  {journal} {\bibinfo  {journal} {Nat. Phys.}\ }\textbf {\bibinfo {volume} {5}},\ \bibinfo {pages} {633} (\bibinfo {year} {2009})}\BibitemShut {NoStop}%
\bibitem [{\citenamefont {Roy}\ \emph {et~al.}(2020)\citenamefont {Roy}, \citenamefont {Chalker}, \citenamefont {Gornyi},\ and\ \citenamefont {Gefen}}]{RoyChalkerGornyiEtAl2020}%
  \BibitemOpen
  \bibfield  {author} {\bibinfo {author} {\bibfnamefont {S.}~\bibnamefont {Roy}}, \bibinfo {author} {\bibfnamefont {J.~T.}\ \bibnamefont {Chalker}}, \bibinfo {author} {\bibfnamefont {I.~V.}\ \bibnamefont {Gornyi}},\ and\ \bibinfo {author} {\bibfnamefont {Y.}~\bibnamefont {Gefen}},\ }\href {https://doi.org/10.1103/PhysRevResearch.2.033347} {\bibfield  {journal} {\bibinfo  {journal} {Phys. Rev. Res.}\ }\textbf {\bibinfo {volume} {2}},\ \bibinfo {pages} {033347} (\bibinfo {year} {2020})}\BibitemShut {NoStop}%
\bibitem [{\citenamefont {Chen}\ \emph {et~al.}(2025)\citenamefont {Chen}, \citenamefont {Kastoryano}, \citenamefont {Brand{\~{a}}o},\ and\ \citenamefont {Gily{\'{e}}n}}]{ChenKastoryanoBrandaoEtAl2025}%
  \BibitemOpen
  \bibfield  {author} {\bibinfo {author} {\bibfnamefont {C.~F.}\ \bibnamefont {Chen}}, \bibinfo {author} {\bibfnamefont {M.}~\bibnamefont {Kastoryano}}, \bibinfo {author} {\bibfnamefont {F.~G.}\ \bibnamefont {Brand{\~{a}}o}},\ and\ \bibinfo {author} {\bibfnamefont {A.}~\bibnamefont {Gily{\'{e}}n}},\ }\href {https://doi.org/10.1038/s41586-025-09583-x} {\bibfield  {journal} {\bibinfo  {journal} {Nature}\ }\textbf {\bibinfo {volume} {646}},\ \bibinfo {pages} {561} (\bibinfo {year} {2025})}\BibitemShut {NoStop}%
\bibitem [{\citenamefont {Ding}\ \emph {et~al.}(2025)\citenamefont {Ding}, \citenamefont {Li},\ and\ \citenamefont {Lin}}]{DingLiLin2025}%
  \BibitemOpen
  \bibfield  {author} {\bibinfo {author} {\bibfnamefont {Z.}~\bibnamefont {Ding}}, \bibinfo {author} {\bibfnamefont {B.}~\bibnamefont {Li}},\ and\ \bibinfo {author} {\bibfnamefont {L.}~\bibnamefont {Lin}},\ }\href {https://doi.org/https://doi.org/10.1007/s00220-025-05235-3} {\bibfield  {journal} {\bibinfo  {journal} {Commun. Math. Phys.}\ }\textbf {\bibinfo {volume} {406}},\ \bibinfo {pages} {1} (\bibinfo {year} {2025})}\BibitemShut {NoStop}%
\bibitem [{\citenamefont {Ding}\ \emph {et~al.}(2024{\natexlab{b}})\citenamefont {Ding}, \citenamefont {Chen},\ and\ \citenamefont {Lin}}]{ding2024single}%
  \BibitemOpen
  \bibfield  {author} {\bibinfo {author} {\bibfnamefont {Z.}~\bibnamefont {Ding}}, \bibinfo {author} {\bibfnamefont {C.-F.}\ \bibnamefont {Chen}},\ and\ \bibinfo {author} {\bibfnamefont {L.}~\bibnamefont {Lin}},\ }\href {https://doi.org/https://doi.org/10.1103/PhysRevResearch.6.033147} {\bibfield  {journal} {\bibinfo  {journal} {Phys. Rev. Res.}\ }\textbf {\bibinfo {volume} {6}},\ \bibinfo {pages} {033147} (\bibinfo {year} {2024}{\natexlab{b}})}\BibitemShut {NoStop}%
\bibitem [{\citenamefont {Li}\ \emph {et~al.}(2025)\citenamefont {Li}, \citenamefont {Zhan},\ and\ \citenamefont {Lin}}]{LiZhanLin2025}%
  \BibitemOpen
  \bibfield  {author} {\bibinfo {author} {\bibfnamefont {H.~E.}\ \bibnamefont {Li}}, \bibinfo {author} {\bibfnamefont {Y.}~\bibnamefont {Zhan}},\ and\ \bibinfo {author} {\bibfnamefont {L.}~\bibnamefont {Lin}},\ }\href {https://doi.org/10.1038/s41534-025-01124-8} {\bibfield  {journal} {\bibinfo  {journal} {npj Quantum Inf.}\ }\textbf {\bibinfo {volume} {11}},\ \bibinfo {pages} {1} (\bibinfo {year} {2025})}\BibitemShut {NoStop}%
\bibitem [{\citenamefont {Lloyd}\ \emph {et~al.}(2025)\citenamefont {Lloyd}, \citenamefont {Michailidis}, \citenamefont {Mi}, \citenamefont {Smelyanskiy},\ and\ \citenamefont {Abanin}}]{Lloyd2025Quasiparticle}%
  \BibitemOpen
  \bibfield  {author} {\bibinfo {author} {\bibfnamefont {J.}~\bibnamefont {Lloyd}}, \bibinfo {author} {\bibfnamefont {A.~A.}\ \bibnamefont {Michailidis}}, \bibinfo {author} {\bibfnamefont {X.}~\bibnamefont {Mi}}, \bibinfo {author} {\bibfnamefont {V.}~\bibnamefont {Smelyanskiy}},\ and\ \bibinfo {author} {\bibfnamefont {D.~A.}\ \bibnamefont {Abanin}},\ }\href {https://doi.org/10.1103/PRXQuantum.6.010361} {\bibfield  {journal} {\bibinfo  {journal} {PRX Quantum}\ }\textbf {\bibinfo {volume} {6}},\ \bibinfo {pages} {010361} (\bibinfo {year} {2025})}\BibitemShut {NoStop}%
\bibitem [{\citenamefont {Zhan}\ \emph {et~al.}(2026)\citenamefont {Zhan}, \citenamefont {Ding}, \citenamefont {Huhn}, \citenamefont {Gray}, \citenamefont {Preskill}, \citenamefont {Chan},\ and\ \citenamefont {Lin}}]{ZhanDingHuhnEtAl2026}%
  \BibitemOpen
  \bibfield  {author} {\bibinfo {author} {\bibfnamefont {Y.}~\bibnamefont {Zhan}}, \bibinfo {author} {\bibfnamefont {Z.}~\bibnamefont {Ding}}, \bibinfo {author} {\bibfnamefont {J.}~\bibnamefont {Huhn}}, \bibinfo {author} {\bibfnamefont {J.}~\bibnamefont {Gray}}, \bibinfo {author} {\bibfnamefont {J.}~\bibnamefont {Preskill}}, \bibinfo {author} {\bibfnamefont {G.~K.-L.}\ \bibnamefont {Chan}},\ and\ \bibinfo {author} {\bibfnamefont {L.}~\bibnamefont {Lin}},\ }\href {https://doi.org/10.1103/wzb3-dbg9} {\bibfield  {journal} {\bibinfo  {journal} {Phys. Rev. X}\ }\textbf {\bibinfo {volume} {16}},\ \bibinfo {pages} {011004} (\bibinfo {year} {2026})}\BibitemShut {NoStop}%
\bibitem [{\citenamefont {Xie}\ \emph {et~al.}(2024)\citenamefont {Xie}, \citenamefont {Xue},\ and\ \citenamefont {Zhang}}]{xie2024variational}%
  \BibitemOpen
  \bibfield  {author} {\bibinfo {author} {\bibfnamefont {X.-D.}\ \bibnamefont {Xie}}, \bibinfo {author} {\bibfnamefont {Z.-Y.}\ \bibnamefont {Xue}},\ and\ \bibinfo {author} {\bibfnamefont {D.-B.}\ \bibnamefont {Zhang}},\ }\href {https://doi.org/https://doi.org/10.1007/s11467-023-1382-3} {\bibfield  {journal} {\bibinfo  {journal} {Front. Phys.}\ }\textbf {\bibinfo {volume} {19}},\ \bibinfo {pages} {41202} (\bibinfo {year} {2024})}\BibitemShut {NoStop}%
\bibitem [{\citenamefont {Zhao}\ \emph {et~al.}(2023)\citenamefont {Zhao}, \citenamefont {Zhang},\ and\ \citenamefont {Wei}}]{zhao2023universal}%
  \BibitemOpen
  \bibfield  {author} {\bibinfo {author} {\bibfnamefont {H.}~\bibnamefont {Zhao}}, \bibinfo {author} {\bibfnamefont {P.}~\bibnamefont {Zhang}},\ and\ \bibinfo {author} {\bibfnamefont {T.-C.}\ \bibnamefont {Wei}},\ }\href {https://doi.org/https://doi.org/10.1038/s41598-023-49662-5} {\bibfield  {journal} {\bibinfo  {journal} {Scientific Reports}\ }\textbf {\bibinfo {volume} {13}},\ \bibinfo {pages} {22313} (\bibinfo {year} {2023})}\BibitemShut {NoStop}%
\bibitem [{\citenamefont {Hatano}\ and\ \citenamefont {Nelson}(1996)}]{hatano1996localization}%
  \BibitemOpen
  \bibfield  {author} {\bibinfo {author} {\bibfnamefont {N.}~\bibnamefont {Hatano}}\ and\ \bibinfo {author} {\bibfnamefont {D.~R.}\ \bibnamefont {Nelson}},\ }\href {https://doi.org/https://doi.org/10.1103/PhysRevLett.77.570} {\bibfield  {journal} {\bibinfo  {journal} {Phys. Rev. Lett.}\ }\textbf {\bibinfo {volume} {77}},\ \bibinfo {pages} {570} (\bibinfo {year} {1996})}\BibitemShut {NoStop}%
\bibitem [{\citenamefont {Hatano}\ and\ \citenamefont {Nelson}(1998)}]{hatano1998non}%
  \BibitemOpen
  \bibfield  {author} {\bibinfo {author} {\bibfnamefont {N.}~\bibnamefont {Hatano}}\ and\ \bibinfo {author} {\bibfnamefont {D.~R.}\ \bibnamefont {Nelson}},\ }\href {https://doi.org/https://doi.org/10.1103/PhysRevB.58.8384} {\bibfield  {journal} {\bibinfo  {journal} {Phys. Rev. B}\ }\textbf {\bibinfo {volume} {58}},\ \bibinfo {pages} {8384} (\bibinfo {year} {1998})}\BibitemShut {NoStop}%
\bibitem [{\citenamefont {Feng}\ \emph {et~al.}(2025)\citenamefont {Feng}, \citenamefont {Liu}, \citenamefont {Zhang},\ and\ \citenamefont {Chen}}]{feng2025numerical}%
  \BibitemOpen
  \bibfield  {author} {\bibinfo {author} {\bibfnamefont {X.}~\bibnamefont {Feng}}, \bibinfo {author} {\bibfnamefont {S.}~\bibnamefont {Liu}}, \bibinfo {author} {\bibfnamefont {S.-X.}\ \bibnamefont {Zhang}},\ and\ \bibinfo {author} {\bibfnamefont {S.}~\bibnamefont {Chen}},\ }\href {https://doi.org/https://doi.org/10.1103/g1cw-tk7f} {\bibfield  {journal} {\bibinfo  {journal} {Phys. Rev. B}\ }\textbf {\bibinfo {volume} {111}},\ \bibinfo {pages} {224310} (\bibinfo {year} {2025})}\BibitemShut {NoStop}%
\bibitem [{\citenamefont {Fefferman}\ and\ \citenamefont {Lin}(2018)}]{fefferman2018complete}%
  \BibitemOpen
  \bibfield  {author} {\bibinfo {author} {\bibfnamefont {B.}~\bibnamefont {Fefferman}}\ and\ \bibinfo {author} {\bibfnamefont {C.~Y.-Y.}\ \bibnamefont {Lin}},\ }in\ \href {https://doi.org/10.4230/LIPIcs.ITCS.2018.4} {\emph {\bibinfo {booktitle} {9th Innovations in Theoretical Computer Science Conference (ITCS)}}}\ (\bibinfo {year} {2018})\ pp.\ \bibinfo {pages} {4--1}\BibitemShut {NoStop}%
\bibitem [{\citenamefont {Kempe}\ and\ \citenamefont {Regev}(2003)}]{kempe20033}%
  \BibitemOpen
  \bibfield  {author} {\bibinfo {author} {\bibfnamefont {J.}~\bibnamefont {Kempe}}\ and\ \bibinfo {author} {\bibfnamefont {O.}~\bibnamefont {Regev}},\ }\href@noop {} {\bibfield  {journal} {\bibinfo  {journal} {arXiv preprint quant-ph/0302079}\ } (\bibinfo {year} {2003})}\BibitemShut {NoStop}%
\bibitem [{\citenamefont {Abramowitz}\ and\ \citenamefont {Stegun}(1965)}]{abramowitz1965handbook}%
  \BibitemOpen
  \bibfield  {author} {\bibinfo {author} {\bibfnamefont {M.}~\bibnamefont {Abramowitz}}\ and\ \bibinfo {author} {\bibfnamefont {I.~A.}\ \bibnamefont {Stegun}},\ }\href@noop {} {\emph {\bibinfo {title} {Handbook of Mathematical Functions: with Formulas, Graphs, and Mathematical Tables}}},\ Vol.~\bibinfo {volume} {55}\ (\bibinfo  {publisher} {Courier Corporation},\ \bibinfo {year} {1965})\BibitemShut {NoStop}%
\bibitem [{\citenamefont {Kliesch}\ \emph {et~al.}(2011)\citenamefont {Kliesch}, \citenamefont {Barthel}, \citenamefont {Gogolin}, \citenamefont {Kastoryano},\ and\ \citenamefont {Eisert}}]{kliesch2011dissipative}%
  \BibitemOpen
  \bibfield  {author} {\bibinfo {author} {\bibfnamefont {M.}~\bibnamefont {Kliesch}}, \bibinfo {author} {\bibfnamefont {T.}~\bibnamefont {Barthel}}, \bibinfo {author} {\bibfnamefont {C.}~\bibnamefont {Gogolin}}, \bibinfo {author} {\bibfnamefont {M.}~\bibnamefont {Kastoryano}},\ and\ \bibinfo {author} {\bibfnamefont {J.}~\bibnamefont {Eisert}},\ }\href {https://doi.org/https://doi.org/10.1103/PhysRevLett.107.120501} {\bibfield  {journal} {\bibinfo  {journal} {Phys. Rev. Lett.}\ }\textbf {\bibinfo {volume} {107}},\ \bibinfo {pages} {120501} (\bibinfo {year} {2011})}\BibitemShut {NoStop}%
\bibitem [{\citenamefont {Cleve}\ and\ \citenamefont {Wang}(2017)}]{cleve2017efficient}%
  \BibitemOpen
  \bibfield  {author} {\bibinfo {author} {\bibfnamefont {R.}~\bibnamefont {Cleve}}\ and\ \bibinfo {author} {\bibfnamefont {C.}~\bibnamefont {Wang}},\ }in\ \href {https://doi.org/https://doi.org/10.4230/LIPIcs.ICALP.2017.17} {\emph {\bibinfo {booktitle} {44th International Colloquium on Automata, Languages, and Programming (ICALP)}}}\ (\bibinfo {year} {2017})\ p.~\bibinfo {pages} {17}\BibitemShut {NoStop}%
\bibitem [{\citenamefont {Childs}\ and\ \citenamefont {Li}(2017)}]{childs2017efficient}%
  \BibitemOpen
  \bibfield  {author} {\bibinfo {author} {\bibfnamefont {A.~M.}\ \bibnamefont {Childs}}\ and\ \bibinfo {author} {\bibfnamefont {T.}~\bibnamefont {Li}},\ }\href {https://doi.org/https://doi.org/10.26421/QIC17.11-12} {\bibfield  {journal} {\bibinfo  {journal} {Quantum Inf. Comput.}\ }\textbf {\bibinfo {volume} {17}},\ \bibinfo {pages} {901} (\bibinfo {year} {2017})}\BibitemShut {NoStop}%
\bibitem [{\citenamefont {Li}\ and\ \citenamefont {Wang}(2023)}]{li2022simulating}%
  \BibitemOpen
  \bibfield  {author} {\bibinfo {author} {\bibfnamefont {X.}~\bibnamefont {Li}}\ and\ \bibinfo {author} {\bibfnamefont {C.}~\bibnamefont {Wang}},\ }in\ \href {https://doi.org/10.4230/LIPIcs.ICALP.2023.87} {\emph {\bibinfo {booktitle} {50th International Colloquium on Automata, Languages, and Programming (ICALP)}}}\ (\bibinfo {year} {2023})\BibitemShut {NoStop}%
\bibitem [{\citenamefont {Ding}\ \emph {et~al.}(2024{\natexlab{c}})\citenamefont {Ding}, \citenamefont {Li},\ and\ \citenamefont {Lin}}]{ding2024simulating}%
  \BibitemOpen
  \bibfield  {author} {\bibinfo {author} {\bibfnamefont {Z.}~\bibnamefont {Ding}}, \bibinfo {author} {\bibfnamefont {X.}~\bibnamefont {Li}},\ and\ \bibinfo {author} {\bibfnamefont {L.}~\bibnamefont {Lin}},\ }\href {https://doi.org/https://doi.org/10.1103/PRXQuantum.5.020332} {\bibfield  {journal} {\bibinfo  {journal} {PRX Quantum}\ }\textbf {\bibinfo {volume} {5}},\ \bibinfo {pages} {020332} (\bibinfo {year} {2024}{\natexlab{c}})}\BibitemShut {NoStop}%
\bibitem [{\citenamefont {Lin}(2022)}]{lin2022lecture}%
  \BibitemOpen
  \bibfield  {author} {\bibinfo {author} {\bibfnamefont {L.}~\bibnamefont {Lin}},\ }\href@noop {} {\bibfield  {journal} {\bibinfo  {journal} {arXiv preprint arXiv:2201.08309}\ } (\bibinfo {year} {2022})}\BibitemShut {NoStop}%
\bibitem [{\citenamefont {Javadi-Abhari}\ \emph {et~al.}(2024)\citenamefont {Javadi-Abhari}, \citenamefont {Treinish}, \citenamefont {Krsulich}, \citenamefont {Wood}, \citenamefont {Lishman}, \citenamefont {Gacon}, \citenamefont {Martiel}, \citenamefont {Nation}, \citenamefont {Bishop}, \citenamefont {Cross}, \citenamefont {Johnson},\ and\ \citenamefont {Gambetta}}]{qiskit2024}%
  \BibitemOpen
  \bibfield  {author} {\bibinfo {author} {\bibfnamefont {A.}~\bibnamefont {Javadi-Abhari}}, \bibinfo {author} {\bibfnamefont {M.}~\bibnamefont {Treinish}}, \bibinfo {author} {\bibfnamefont {K.}~\bibnamefont {Krsulich}}, \bibinfo {author} {\bibfnamefont {C.~J.}\ \bibnamefont {Wood}}, \bibinfo {author} {\bibfnamefont {J.}~\bibnamefont {Lishman}}, \bibinfo {author} {\bibfnamefont {J.}~\bibnamefont {Gacon}}, \bibinfo {author} {\bibfnamefont {S.}~\bibnamefont {Martiel}}, \bibinfo {author} {\bibfnamefont {P.~D.}\ \bibnamefont {Nation}}, \bibinfo {author} {\bibfnamefont {L.~S.}\ \bibnamefont {Bishop}}, \bibinfo {author} {\bibfnamefont {A.~W.}\ \bibnamefont {Cross}}, \bibinfo {author} {\bibfnamefont {B.~R.}\ \bibnamefont {Johnson}},\ and\ \bibinfo {author} {\bibfnamefont {J.~M.}\ \bibnamefont {Gambetta}},\ }\href {https://doi.org/10.48550/arXiv.2405.08810} {\bibinfo {title} {Quantum computing with {Q}iskit}} (\bibinfo {year} {2024}),\ \Eprint {https://arxiv.org/abs/2405.08810} {arXiv:2405.08810 [quant-ph]}
  \BibitemShut {NoStop}%
\end{thebibliography}%

\newpage
\widetext

\newpage

 \begin{center}
    \textbf{SUPPLEMENTARY MATERIALS}
 \end{center}

\appendix
\section{Computational complexity of non-Hermitian local Hamiltonians}
\label{append:hardness-non-hermitian-hamiltonian}

Recall that an $n$-qubit quantum operator $A = \sum_j c_j H_j$ is $k$-\textit{local} if each $H_j$ is a Hermitian operator acting on at most $k$ qubits. Here, the coefficients $c_j$ can be complex-valued so $A$ can be non-Hermitian.
For a matrix $A$, its $\epsilon$-pseudospectrum consists of all complex numbers $z$ at which the matrix resolvent (i.e., $(A-zI)^{-1}$) has a norm no less than $1/\epsilon$:
\[\Lambda_\epsilon(A) \coloneqq \{z \in \CC \colon \|(A-zI)^{-1}\| \ge 1/\epsilon\}.\]

\begin{problem}[$\epsilon$-Pseudospectrum]\label{prob:pseudo-spectrum}
    Given a number $z \in \CC$, a tolerance $\epsilon = \Omega(1/\poly(n))$ and a $k$-local operator $A = \sum_j c_j H_j$ such that for every $j$: $c_j \in \CC$, $|c_j| \le \poly(n)$, and $0 \preceq H_j \preceq 1$. We assume that the integer $k$, complex number $z$, tolerance parameter $\epsilon$, and coefficients $c_j$ are specified with $\poly(n)$ bits of precision. 
    The problem is to decide, under the promise that one of the following holds: \textbf{(YES)} $\sigma_0(A-zI) \le \epsilon$, \textbf{(NO)} $\sigma_0(A-zI) \ge 2\epsilon$.
\end{problem}

\begin{problem}[$\epsilon$-Neighborhood of Spectrum]\label{prob:distance-to-spec}
    Given a number $z \in \CC$, a tolerance $\epsilon = \Omega(1/\poly(n))$ and a $k$-local operator $A = \sum_j c_j H_j$ such that for every $j$: $c_j \in \CC$, $|c_j| \le \poly(n)$, and $0 \preceq H_j \preceq 1$. Let $\{\lambda_k\}_{k}$ be the eigenvalues of $A$. 
    We assume that the integer $k$, complex number $z$, tolerance parameter $\epsilon$, and coefficients $c_j$ are specified with $\poly(n)$ bits of precision.  
    The problem is to decide, under the promise that one of the following holds: \textbf{(YES)} $\min_k |z - \lambda_k| \le \epsilon$, \textbf{(NO)} $\min_k |z - \lambda_k| \ge 2\epsilon$.
\end{problem}

Our main technical results are stated as follows.

\begin{theorem}\label{thm:pseudo-spectrum-qma}
    For any $4 \le k \le \mathcal{O}(1)$, \cref{prob:pseudo-spectrum} is $\QMA$-complete.
\end{theorem}

\begin{theorem}\label{thm:distance-to-spec-pspace}
    For any $5 \le k \le \mathcal{O}(1)$, \cref{prob:distance-to-spec} is $\PSPACE$-hard.
\end{theorem}

\subsection{Proof of~\texorpdfstring{\cref{thm:pseudo-spectrum-qma}}{}}\label{append:proof-qma-hardness}

\begin{problem}[$k$-local Hamiltonian]\label{prob:k-local-ham}
    Given as input is a $k$-local Hamiltonian $H = \sum^r_{j=1} H_j$ acting on $n$ qubits satisfying $0 \preceq H_j \preceq 1$ and $r = \poly(n)$, and numbers $a, b$ satisfying $b-a = \Omega(1/\poly(n))$. It is promised that the smallest eigenvalue of $H$ is either at most $a$ or at least $b$. Output \textbf{(YES)} if the smallest eigenvalue of $H$ is at most $a$, or $\textbf{(NO)}$ if the smallest eigenvalue of $H$ is at least $b$.
\end{problem}

In what follows, we show~\cref{prob:pseudo-spectrum} is $\QMA$-complete by making connections to the $k$-local Hamiltonian (\cref{prob:k-local-ham}), a classic $\QMA$-complete problem.

\begin{proof}[Proof of~\cref{thm:pseudo-spectrum-qma}]
    To see~\cref{prob:pseudo-spectrum} belongs to $\QMA$, we use the identity:
    \[\sigma^2_0(A_z) = \lambda_0(H_z),\quad H_z = A^\dagger_z A_z.\]
    Since $A$ is a $k$-local operator, $H_z$ is $2k$-local. Deciding if $\sigma_0(A_z) \le \epsilon$ (YES) or $\sigma_0(A_z) \ge 2\epsilon$ (NO), now reduces to deciding whether $\lambda_0(H_z) \le a = \epsilon^2$ or $\lambda_0(H_z) \ge b=4\epsilon^2$.

    The hardness part relies on reducing any $k$-local Hamiltonian to an $\epsilon$-pseudospectrum problem.
    Given a $k$-local Hamiltonian $H$ and $0 \le a < b$ with $b - a = \Omega(1/\poly(n))$. Without loss of generality, we assume $a \le \|H\| = \poly(n)$. We denote $\lambda_0(H)$ as the minimal eigenvalue of this Hamiltonian $H$.
    
    \textbf{Case 1: $b > 2a$.}
    In this case, the local Hamiltonian can be reduced directly to an $\epsilon$-pseudospectrum problem. If $a>0$, we choose $A = H$, $z = 0$ and $\epsilon = a$. Then if $\lambda_0(H) \in [0, a]$, we have $\sigma_0(H-zI) = |\lambda_0| \le a = \epsilon$, while if $\lambda_0(H) \in [b, \infty)$, then $\sigma_0(H-zI) = |\lambda_0| \ge b > 2a = 2\epsilon$. If instead $a=0$, then the YES case forces $\lambda_0(H)=0$ because $H \succeq 0$. In that subcase we choose $A=H$, $z=0$, and $\epsilon=b/2$. Then the YES case gives $\sigma_0(H)=0 \le \epsilon$, while in the NO case $\sigma_0(H)=\lambda_0(H) \ge b = 2\epsilon$.
    
    \textbf{Case 2: $b \le 2a$.}
    In this case, we construct a new Hamiltonian:
    \begin{equation}\label{eqn:hardness-qma-A}
        A = \sum^m_{j=1}\ketbra{j}{j}\otimes (H-z_j),
    \end{equation}
    where $m = \lceil \frac{3a}{b-a}\rceil$, and $\{z_j:= \frac{(b-a)j}{3}\}^m_{j=1}$ forms a mesh grid on $[0, a]$ with an $\epsilon:= \frac{(b-a)}{3}$ distance between two adjacent points. If $\lambda_0(H) \in [0,a]$, the distance between $\lambda_0(H)$ and the nearest grid point is at most $\epsilon$, meaning $\sigma_0(A) \le \epsilon$; while if $\lambda_0(H) \ge b$, the smallest distance between the spectrum of $H$ and the grid points $\{z_j\}^m_{j=1}$ is at least $(b-a)-\epsilon \ge 2\epsilon$. Therefore, by deciding if $0$ is in the $\epsilon$-pseudospectrum of $A$, we solve the $k$-local Hamiltonian problem.
    
    Note that the locality of the new matrix $A$ as in~\cref{eqn:hardness-qma-A} is $k+\log_2(m)$, since $m= \poly(n)$, the locality of $A$ scales as $\mathcal{O}(\log(n))$. To further reduce the locality of $A$ to a constant, we employ the standard ``unary clock'' technique from Hamiltonian complexity theory~\cite{kitaev2002classical}.
    Consider the following Hamiltonian operator, 

    \begin{equation}\label{eqn:gadget-hamiltonian-with-penalty}
        \AA = \AA_0 + \gamma \HH_{\rm pen},\quad \AA_0 = \sum^m_{j=1}\hat{n}_{j} \otimes (H-z_j),\quad \HH_{\rm pen} = \left(\sum^{m}_{j=1}\hat{n}_j - 1\right)^2 \otimes I,
    \end{equation}
    where $\hat{n}_j = (I - \sigma^{(z)}_j)/2$ is the number operator acting on the $j$-th ancilla qubit, and $\gamma > 0$ is a penalty parameter. 
    Here, the ancilla register uses the one-hot codeword to encode the computational basis $\{\ket{j}\}^m_{j=1}$. The operator $\AA$ is $(k+2)$-local acting on $n+m = \poly(n)$ qubits. Let $\SS$ be the encoding subspace spanned by the one-hot codewords, and we have $A = \AA|_\SS$. 
    Since both $\AA_0$ and $\HH_{\rm pen}$ are polynomials in the commuting clock number operators, $\AA$ is block diagonal with respect to the Hamming-weight decomposition of the ancilla register. For a subset $S \subseteq \{1,\dots,m\}$ with $|S|=r$, the restriction of $\AA$ to the sector where exactly the qubits in $S$ are occupied is
    \[
        \AA_S = rH - \left(\sum_{j\in S} z_j\right)I + \gamma(r-1)^2 I.
    \]
    In particular, when $r=1$ we recover the desired block $H-z_j$, so $\AA|_\SS = A$. When $r=0$, the operator is simply $\gamma I$. When $r\ge 2$, every eigenvalue of $\AA_S$ is bounded below by
    \[
        \gamma(r-1)^2 - \sum_{j\in S} z_j \ge \gamma - (ma+\epsilon),
    \]
    because $H \succeq 0$ and $0 \le z_j \le a$ for every $1\le j\le m-1$ and $0\le z_m \le a+\epsilon$. We therefore choose
    \[
        \gamma := ma + 3\epsilon = \poly(n),
    \]
    and every eigenvalue in $\SS^\perp$ is at least $2\epsilon$ away from the origin. Since $\AA$ is Hermitian, its singular values are the absolute values of its eigenvalues, and hence
    \[
        \sigma_0(\AA) = \min\{\sigma_0(A),\sigma_0(\AA|_{\SS^\perp})\},\qquad \sigma_0(\AA|_{\SS^\perp}) \ge 2\epsilon.
    \]
    Therefore the two promise cases for $A$ are preserved exactly by $\AA$: (YES) $\sigma_0(\AA) \le \epsilon$, and (NO) $\sigma_0(\AA) \ge 2\epsilon$. Reducing from $2$-local Hamiltonian yields a $4$-local instance, so \cref{prob:pseudo-spectrum} is $\QMA$-complete for any locality at least $4$.
\end{proof}

\subsection{Proof of~\texorpdfstring{\cref{thm:distance-to-spec-pspace}}{}}\label{append:exact-problem-pspace}

\begin{problem}[Precise $k$-local Hamiltonian]
    Given as input is a $k$-local Hamiltonian $H = \sum^r_{j=1} H_j$ acting on $n$ qubits, satisfying $r \in \poly(n)$ and $\|H_j\| \le \poly(n)$, and numbers $a , b$ satisfying $b - a > 2^{-\poly(n)}$. It is promised that the smallest eigenvalue of $H$ is either at most $a$ or at least $b$. Output \textbf{(YES)} if the smallest eigenvalue of $H$ is at most $a$, and output $\textbf{(NO)}$ otherwise.    
\end{problem}

\begin{theorem}[Theorem 24,~\cite{fefferman2018complete}]\label{thm:pspace-precise-local-hamiltonian}
    Precise $3$-local Hamiltonian is $\PSPACE$-complete. 
\end{theorem}
\begin{proof}
    The proof of~\cref{thm:pspace-precise-local-hamiltonian} is available in~\cite[Appendix~E]{fefferman2018complete}. 
    For completeness, we provide the reduction details here, as they will be useful in our proof of~\cref{thm:distance-to-spec-pspace} later.
    For a $\QMA$-verification procedure with $T = \poly(n)$ gates, completeness $c$ and soundness $s$, Kempe and Regev~\cite{kempe20033} show that it can be reduced to a $3$-local Hamiltonian with lowest eigenvalue no more than $(1-c)/(T+1)$ in the YES case, and no less than $(1-s)/T^3$ in the NO case. For $\PreciseQMA$-hard problems, perfect completeness can be assumed~\cite[Theorem 25]{fefferman2018complete}. So it suffices to take $c = 1$ and $s = 1 - 2^{-\poly(n)}$, and the $\QMA$-verification procedure is reduced to a $3$-local Hamiltonian $H$ with the lowest eigenvalue $\lambda_0$ promised to be in either: (YES) $\lambda_0 = 0$ or (NO) $\lambda_0 \ge 2^{-p(n)}/T^3$, where $p(n)$ is a polynomial function in $n$.
\end{proof}

The key idea of our proof is to use a non-Hermitian ``gadget'' to exponentially amplify the spectral gap in a precise $k$-local Hamiltonian. To see how it works, we consider the following $m$-by-$m$ matrix (where $m \ge 1$ is an arbitrary integer and $\delta > 0$):
    \[G =\begin{bmatrix}
    0 & 1 & 0 & \cdots & 0 \\
    0 & 0 & 1 & \cdots & 0 \\
    \vdots &  & \ddots & \ddots & \vdots \\
    0 & \cdots & 0 & 0 & 1 \\
    \delta & 0 & \cdots & 0 & 0
    \end{bmatrix}.\]
The eigenvalue equation $\det(\lambda I - G) = 0$ can be solved analytically:
\[\lambda^m = \delta \implies \lambda_k = e^{2\pi i k/m}\delta^{1/m},\quad k = 0,\dots, m-1.\]

The eigenvalues of $G$ are located on a circle with radius $\delta^{1/m}$.
This means that, even if $\epsilon$ is small, for example, $\delta = 2^{-m}$, by embedding it into the lower left corner of this non-Hermitian gadget matrix $G$, its magnitude can be amplified to a constant $\delta^{1/m} = 1/2$ by evaluating the spectrum of $G$.
In what follows, we generalize this observation for any non-negative Hermitian matrix.

\begin{lemma}\label{lem:spectral-amplify-gadget}
    Suppose $H$ is an $N$-by-$N$ non-negative matrix with eigenvalues $\{\mu_j\}^N_{j=1}$, arranged in an ascending order (counting multiplicity). Given an integer $m\ge 1$, we define the following $Nm$-by-$Nm$ matrix:
    \begin{equation}\label{eqn:non-hermitian-gadget}
        \GG = \begin{bmatrix}
    0 & I & 0 & \cdots & 0 \\
    0 & 0 & I & \cdots & 0 \\
    \vdots &  & \ddots & \ddots & \vdots \\
    0 & \cdots & 0 & 0 & I \\
    H & 0 & \cdots & 0 & 0
    \end{bmatrix},
    \end{equation}
where each $I$ is the $N$-by-$N$ identity matrix. Then, the following holds:
\begin{enumerate}
    \item If $\mu_0 = 0$ (i.e., $H$ is not strictly positive), the matrix $\GG$ has a zero eigenvalue.
    \item If $\mu_0 > 0$ (i.e., $H$ is positive), the matrix $\GG$ is diagonalizable, and its eigenvalues are 
    \begin{equation}\label{eqn:gadget-eigenvalue}
        \lambda_{j,k} = e^{2\pi i k/m}\mu^{1/m}_j,\quad 1 \le j \le N, \quad 0 \le k \le m-1.
    \end{equation}
\end{enumerate}
\end{lemma}

\begin{proof}
    Let $\lambda$ be an eigenvalue of $\GG$ with a corresponding eigenvector $v \in \CC^{Nm}$. We can partition $v$ into $m$ blocks: $v = [v^\top_1, \dots, v^\top_m]^\top$, and the eigenvalue equation can be written as the following system:
    \begin{align*}
    \begin{cases}
        v_2 &= \lambda v_1,\\
        v_3 &= \lambda v_2,\\
        &\vdots\\
        v_m &= \lambda v_{m-1},\\
        H v_1 &= \lambda v_m.
    \end{cases}
    \end{align*}
    By induction, we have $v_{k+1} = \lambda^{k}v_1$ for $k = 1,\dots, m-1$. 
    Substituting this relation into the last equation in the system, we get $Hv_1 = \lambda^m v_1$.
    For $v$ to be an eigenvector of $\GG$, $v_1 \neq 0$. This means $v_1$ has to be an eigenvector of $H$, and all eigenvalues of $\GG$ are given by the roots of
    \[\lambda^m  = \mu_k,\quad \mu_k \in \spec(H).\]

    \textbf{Case 1: $\mu_0 = 0$.}
    Then $\GG$ has an eigenvalue $\lambda = 0$. However, we remark that the matrix $\GG$ may not be diagonalizable since the $0$-eigenvalue block can be defective.
    
    \textbf{Case 2: $\mu_0 > 0$.}
    For positive $H$, the above equation has $m$ distinct solutions for each eigenvalue $\mu_j > 0$. So, for each eigenvector of $H$, we can construct $m$ distinct eigenvectors of $\GG$ as shown above. Since $H$ is Hermitian, the set of $Mm$ constructed eigenvectors of $\GG$ is linearly independent. Therefore, $\GG$ is diagonalizable, and its eigenvalues are as in~\cref{eqn:gadget-eigenvalue}.
\end{proof}

Now, we are ready to prove that~\cref{prob:distance-to-spec} is $\PSPACE$-hard for $k = \mathcal{O}(1)$.

\begin{proof}[Proof of \cref{thm:distance-to-spec-pspace}]
    Following the proof of~\cref{thm:pspace-precise-local-hamiltonian}, any $\PSPACE$ problem can be reduced to a $3$-local Hamiltonian whose ground energy is $0$ in the YES case, and at least $2^{-p(n)}/T^3$ in the NO case, where $p(n)$ is a polynomial and $T = \poly(n)$.

    Choose
    \[
        m := p(n) + 3\lceil \log_2 T\rceil.
    \]
    By~\cref{lem:spectral-amplify-gadget}, we construct the matrix
    \[\GG = \sum^{m-1}_{t=1}\ketbra{t-1}{t}\otimes I + \ketbra{m-1}{0}\otimes H,\]
    which is a $(3+\log_2(m))$-local operator acting on $n+\log_2(m)$ qubits. 
    In the (YES) case, $\GG$ has a zero eigenvalue. In the (NO) case, all the eigenvalues of $\GG$ have magnitude no smaller than 
    \[
        \left(\frac{2^{-p(n)}}{T^3}\right)^{1/m} \ge \left(2^{-p(n)-3\lceil \log_2 T\rceil}\right)^{1/m} = \frac{1}{2}.
    \]
    Therefore, we can solve the precise $k$-local Hamiltonian problem by determining if $z=0$ is in the $\epsilon$-neighborhood of the spectrum of $\GG$ with $\epsilon=1/8$.
    However, the locality of $\GG$ still grows with $n$: $3 + \log_2(m) = \mathcal{O}(\log(n))$, so we again use the unary clock. Using the same one-hot code subspace as in the proof of~\cref{thm:pseudo-spectrum-qma}, the hopping term $\ketbra{t-1}{t}$ can be encoded as a $2$-local operator $\sigma^-_t\otimes \sigma^+_{t-1}$, with the qubit creation and annihilation operators ($\sigma^x$, $\sigma^y$ are Pauli-X and Pauli-Y operators):
    \begin{equation}
        \sigma^+_t = \frac{1}{2}\left(\sigma^x_t -i\sigma^y_t\right),\quad \sigma^-_t = \frac{1}{2}\left(\sigma^x_t +i\sigma^y_t\right).
    \end{equation}
    The resulting gadget operator is
    \[
        \GG_0 := \sum_{t=1}^{m-1}\sigma^-_t\sigma^+_{t-1}\otimes I + \sigma^-_0\sigma^+_{m-1}\otimes H,
    \]
    which is $5$-local and acts on $n+m = \poly(n)$ qubits.
    By adding a penalty term $\HH_{\rm pen}$ as in~\cref{eqn:gadget-hamiltonian-with-penalty} to $\GG_0$, we obtain the full gadget Hamiltonian:
    \begin{equation}
        \widetilde{\GG} = \GG_0 + \gamma \HH_{\rm pen},\qquad \HH_{\rm pen} = \left(\sum_{t=0}^{m-1}\hat{n}_t - 1\right)^2\otimes I.
    \end{equation}
    Let $\SS_r$ denote the clock subspace with exactly $r$ excitations. Each term in $\GG_0$ preserves excitation number, so both $\GG_0$ and $\widetilde{\GG}$ are block diagonal with respect to $\bigoplus_{r=0}^m \SS_r$. On the one-hot sector $\SS_1$, the operator $\GG_0$ coincides exactly with $\GG$, and $\HH_{\rm pen}$ vanishes, so $\widetilde{\GG}|_{\SS_1} = \GG$. Define
    \[
        R := (m-1) + \sum_j \|H_j\|,
    \]
    where $H = \sum_j H_j$. By the triangle inequality, $\|\GG_0\| \le R$, so every restriction $\GG_0|_{\SS_r}$ also has norm at most $R$. For $r \ne 1$, the penalty contributes the scalar $\gamma(r-1)^2$ on $\SS_r$. Choosing
    \[
        \gamma := R + 1,
    \]
    every eigenvalue $\lambda$ of $\widetilde{\GG}|_{\SS_r}$ with $r \ne 1$ obeys
    \[
        |\lambda| \ge \gamma(r-1)^2 - \|\GG_0|_{\SS_r}\| \ge 1.
    \]
    Hence the complement sectors are uniformly separated from the origin, while the one-hot sector retains exactly the spectrum of $\GG$.
    
    Consequently, in the YES case $0 \in \spec(\widetilde{\GG})$, while in the NO case every eigenvalue in the one-hot sector has magnitude at least $1/2$ and every eigenvalue in the complement has magnitude at least $1$. Therefore $0$ is in the $1/8$-neighborhood of $\spec(\widetilde{\GG})$ in the YES case and is at distance at least $1/2$ from $\spec(\widetilde{\GG})$ in the NO case. This proves that~\cref{prob:distance-to-spec} is $\PSPACE$-hard for any $k\ge 5$.
\end{proof}

\section{Quantum singular-value Gaussian-filtered search (QSIGS)}\label{append:qsigs}

\subsection{Algorithmic overview}\label{sec:qsigs}

QSIGS leverages a Gaussian filter-based signal processing approach to compute the singular values of a matrix $A$. 
We assume that $A$ has a singular value decomposition $A = \sum^{M-1}_{m=0} \sigma_m \ketbra{u_m}{v_m}$, where $\{\ket{u_m}\}^{M-1}_{m=0}$, $\{\ket{v_m}\}^{M-1}_{m=0}$ are the left and right singular vectors of $A$, and $\{\sigma_m\}^{M-1}_{m=0}$ are the singular values of $A$ (arranged in an ascending order, counting multiplicity). 
Without loss of generality, we assume all the singular values of $A$ are upper bounded by $1$, i.e., $0 \le \sigma_m \le 1$ for all $0 \le m \le M-1$.  For some $t \ge 0$, we consider applying a $\sin$ transformation on the singular values of $A$: 
\begin{equation}
    \sinsv(tA) \coloneqq \sum^{M-1}_{m=0}\sin(t\sigma_m)\ketbra{u_m}{v_m}.
\end{equation}
Given access to a block-encoding of $A$, this transformation can be implemented using Quantum Singular Value Transformation (QSVT). The output is a unitary operator $U$ that encodes the matrix $\sinsv(tA)$ in the upper left corner, i.e., $(\bra{0^r} \otimes I) U (\ket{0^r}\otimes I) = \sinsv(tA)$, where $r$ is the number of ancilla qubits used in the block-encoding of $A$.
If we apply this unitary operator to a state $\ket{0^r}\otimes \ket{\psi}$ with $\ket{\psi} = \sum^{M-1}_{m=0} c_m \ket{v_m}$ and measure the ancilla register (with $r$ qubits) of $U\ket{\psi}$, the probability of obtaining an all-zero outcome (i.e., $0^r$) can be evaluated by the Born rule:
\begin{equation}\label{eqn:p0t}
    P_0(t) \coloneqq \left\|\sinsv(tA)\ket{\psi}\right\|^2 = \sum^{M-1}_{m=0} p_m \sin^2(\sigma_m t) = \frac{1}{2} - \frac{1}{2}\sum^{M-1}_{m=0} p_m \cos(2\sigma_m t),\quad p_m = |c_m|^2.
\end{equation}
Now, let $t_n$ be i.i.d. samples from a probability distribution $a(t)$, we define a random variable $Z_n$ such that
\begin{equation}
    Z_n = \begin{cases}
        -1, &\text{Measurement outcome is $0^r$;}\\
        1, &\text{Otherwise}.
    \end{cases}
\end{equation}
For a fixed $t$, $Z_n$ is a bounded and $\mathbb{E}[Z_n] = 1-2P_0(t)$.
Given a sample number $N \ge 1$, we define the following filter function for $\theta \in [0,1]$:
\begin{equation}\label{eqn:filter}
    F(\theta) := \frac{1}{N} \sum^N_{n=1} Z_n \exp(-2i \theta t_n).
\end{equation}
For a large enough $N$, the value of the filter function $F(\theta)$ can be approximated by its expectation:
\begin{align}\label{eqn:filter-a}
    \mathbb{E}[F(\theta)] &= \frac{1}{N}\sum_{n=1}^N \mathbb{E}_{t_n,Z_n}\left[Z_n \exp(-2i\theta t_n)\right] = \frac{1}{2}\sum^{M-1}_{m=0} p_m \mathbb{E}_{t\sim a}\left[\exp(2it(\sigma_m - \theta)) + \exp(-2it(\sigma_m+\theta))\right].
\end{align}
If we choose a Gaussian distribution $a(t) = \frac{1}{\sqrt{2\pi}T}\exp(\frac{-t^2}{2T^2})$, its characteristic function is also Gaussian: $\hat{a}(\xi) = \mathbb{E}_{t\sim a}[e^{it\xi}] = \exp(-T^2\xi^2/2)$.
Plugging the characteristic function $\hat{a}(\xi)$ into~\cref{eqn:filter-a}, we have
\begin{equation}\label{eqn:filter-b}
    F(\theta) \approx \frac{1}{2}\sum^{M-1}_{m=0} p_m \left[e^{-2T^2(\sigma_m-\theta)^2} + e^{-2T^2(\sigma_m + \theta)^2}\right].
\end{equation}
Therefore, for a dominant $p_0$ and large enough $T$, the filter function $F(\theta)$ should peak around the ground singular value $\sigma_0$ (note that we have restricted $\theta \in [0,1]$).

To identify the peak of $|F(\theta)|$ over the numerical domain $[0,1]$, we define the set of candidates, denoted by $\Theta = \{\theta_j\}^J_{j=0}$, where
\begin{equation}
    \theta_j = \frac{jq}{T}, \quad 0 \le j \le J \coloneqq \lfloor T/q\rfloor,
\end{equation}
and $q>0$ is a granularity parameter. Using the data pairs $\{(t_n, Z_n)\}^N_{n=1}$, we can evaluate the filter function $F(\theta)$ over the numerical grid $\Theta$. 
The maximizer of the function $|F(\theta)|$ in $\theta \in \Theta$ is provably a good approximation of the ground singular value $\sigma_0$ for sufficiently large $T$ and $N$.

A pseudocode of the entire QSIGS procedure is summarized in~\cref{algo:qsigs}, and an illustrative quantum circuit for generating the data pairs $(t_n, Z_n)$ is shown in~\cref{fig:qsigs_qc}.

\begin{algorithm}[ht!]
\caption{Quantum Singular Value Gaussian Filtered Search}\label{algo:qsigs}
\KwData{Number of data pairs: $N$; sampling distribution $a_T(t)$; block-encoding $O_A$ of $A$; initial state $\rho$; set of candidates $\Theta = \{\theta_j:=jq/T\}_{j=0}^J$.}
\KwResult{$\theta^*$ as an estimate of the ground singular value $\sigma_0$.}

Dataset $d \gets \{\}$\;
Sample $\{t_n\}_{n=1}^N$ independently from $a_T(t)$\;
\For{$n=1,\dots,N$}{
    Construct a quantum circuit $\sin^{\mathrm{SV}}(t_n A)$ using~\cref{lem:sin-cost}\;
    
    Prepare $\ketbra{0^{r}}{0^r}\otimes \rho$ and apply $\sin^{\mathrm{SV}}(t_n A)$\;
    
    Measure the ancilla register and record the outcome $s_n$\;
    
    \eIf{$s_n = 0^{r}$}{
        $Z_n \gets -1$\;
    }{
        $Z_n \gets 1$\;
    }
    
    Append $(t_n, Z_n)$ to $d$\;
}

$F(\theta) \gets \frac{1}{N}\sum_{n=1}^N Z_n \exp(-2i \theta t_n)$ for $\theta \in \Theta$ and data pairs in $d$\;
$\theta^* \gets \arg\max_{\theta \in \Theta} |F(\theta)|$\;
\end{algorithm}

\begin{figure}[!ht]
    \centering
    \includegraphics[width=0.6\linewidth]{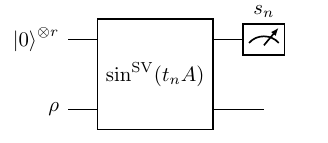}
    \caption{An illustrative quantum circuit for generating the data pairs $(t_n, Z_n)$ in~\cref{algo:qsigs}. The gate here is a block-encoding of $\sin^{\mathrm{SV}}(t_n A)$, which can be implemented using QSVT. The measurement outcome $Z_n$ is determined by whether the ancilla register is in the all-zero state or not.
    }
    \label{fig:qsigs_qc}
\end{figure}

\subsection{Complexity analysis}

For the input model, we adopt the commonly used block-encoding which is a unitary that encodes the information of interest in the sub-blocks. For a general matrix $A$, we call $U_A$ an $(\alpha,a,\epsilon)$-block-encoding of $A$ if
\begin{equation}
    \|A - \alpha(\bra{0^a}\otimes I) U_A (\ket{0^a}\otimes I)\| \leq \epsilon.
\end{equation}
It is clear that $\alpha$ should attain the property that $\alpha \geq \|A\|$ since $A/\alpha$ is stored in a unitary.

The data generating process is for creating a dataset which can be used for singular value estimation. Starting from the filter function, we need to have a formal definition on the truncated Gaussian density function. For a given $T$, and some parameter $\sigma$, we define
\begin{equation}
    a_T^{\text{truncate}}(t)
    =
    \left(1 - \int_{-\sigma T}^{\sigma T} \frac{1}{\sqrt{2\pi }T}\exp\left(\frac{-s^2}{2T^2}\right)\mathbf{1}_{[-\sigma T, \sigma T]} \mathrm{d}s\right)\delta_0(t)
    + \frac{1}{\sqrt{2\pi }T}\exp\left(\frac{-t^2}{2T^2}\right)\mathbf{1}_{[-\sigma T, \sigma T]}.
\end{equation}
The distribution above is simply truncating the possibilities of sampling $t$'s whose absolute values are greater than $\sigma T$.
One may tune the parameter $\sigma$ to adjust the possible max runtime.

After sampling $N$ samples $\{t_n\}_{n=1}^N$ from this distribution, our next step is about implementing the block-encodings of $\{\sin^{\text{SV}}(A t_n)\}_{n=1}^N$ base on the block-encoding of $A$. This is nothing but a simple call of QSVT, which can be detailed as follows:
\begin{lemma}[Cost of implementing $\sin^{\text{SV}}$ transformation]
\label{lem:sin-cost}
    Given an $(1, a, 0)$ block-encoding $O_A$ that encodes $A$,
    a time parameter $t$
    and an error term $\epsilon$, we can implement a unitary $U$ which is a block-encoding of $\sin^{\textrm{SV}}(A t)$ up to error $\epsilon$ with $\Theta\left(|t| + \frac{\log(1/\epsilon)}{\log\log(1/\epsilon)}\right)$ queries to $O_A$.
\end{lemma}
\begin{proof}
    Note that the $\sin$ function in $[-1,1]$ can be approximated within error $\epsilon$ by an odd polynomial of degree $R$ where
    $R = \Theta\left(|t| + \frac{\log(1/\epsilon)}{\log\log(1/\epsilon)}\right)$~\cite{gilyen2019quantum} using the Jacobi-Anger expansion~\cite{abramowitz1965handbook}. Using Quantum Singular Value Transformation allows us to construct the desired block-encoding with $\Theta(R) = \Theta\left(|t| + \frac{\log(1/\epsilon)}{\log\log(1/\epsilon)}\right)$ queries to $O_A$ or $O_A^\dag$.
\end{proof}

Recall that the filter function is $F(\theta) = \frac{1}{N}\sum^N_{n=1}Z_n \exp(-2i\theta t_n)$. In the following lemma, we characterize the empirical error due to the finite sample size and the truncated Gaussian filter.

\begin{lemma}[Empirical error]\label{lem:empirical-error}
    Given an initial state $\rho$, let $\{(t_n,Z_n)^N_{n=1}\}$ be the data pairs collected in~\cref{algo:qsigs}.
    Let $\{\ket{v_m}\}^{M-1}_{m=0}$ be the right singular vectors of $A$, arranged by ascending singular values. Denote $p_m := \Tr[\ketbra{v_m}{v_m}\rho]$.
    We define 
    \begin{equation}
        E(\theta) := \frac{1}{N} \sum^N_{n=1} Z_n \exp(-2i \theta t_n) - \frac{1}{2}\sum^{M-1}_{m=0} p_m \left[e^{-2T^2(\sigma_m-\theta)^2} + e^{-2T^2(\sigma_m + \theta)^2}\right].
    \end{equation}
    \begin{itemize}
        \item Given an arbitrary $\delta > 0$, by choosing $\sigma = \mathcal{O}(\log^{1/2}(1/\delta))$, we have 
        \begin{equation}\label{eqn:truncation-error}
            \left|\mathbb{E}\left[\frac{1}{N} \sum^N_{n=1} Z_n \exp(-2i \theta t_n)\right] - \frac{1}{2}\sum^{M-1}_{m=0} p_m \left[e^{-2T^2(\sigma_m-\theta)^2} + e^{-2T^2(\sigma_m + \theta)^2}\right]\right| \le \delta/2.
        \end{equation}
        \item Given an arbitrary $\eta > 0$, and $\Theta = \{\theta_j\}^J_{j=0}$. By choosing $N = \mathcal{O}\left(\frac{\log(J/\eta)}{\delta^2}\right)$, we have
        \begin{equation}\label{eqn:success-prob-empirical}
            \mathbb{P}\left[\max_{\theta \in \Theta}|E(\theta)| \ge \delta \right] \le \eta.
        \end{equation}
    \end{itemize}
\end{lemma}
\begin{proof}
The proof of~\cref{eqn:truncation-error} follows a similar line as in~\cite[Lemma A.1, Eq. (19)]{ding2024quantum}. 
For a fixed $\theta \in [0,1]$, define
\[
    \widetilde{E}(\theta) := \frac{1}{N} \sum_{n=1}^N Z_n e^{-2i\theta t_n} - \mathbb{E}\left[Z_n e^{-2i\theta t_n}\right].
\]
Then both $\Re X_n(\theta)$ and $\Im X_n(\theta)$ lie in $[-2,2]$, where $X_n(\theta) := Z_n e^{-2i\theta t_n} - \mathbb{E}[Z_n e^{-2i\theta t_n}]$. Applying Hoeffding's inequality separately to the real and imaginary parts and using the union bound, we obtain
\begin{equation}
\mathbb{P}[|\widetilde{E}(\theta)| \ge \delta/2] \le \mathbb{P}[|\Re \widetilde{E}(\theta)| \ge \delta/4] + \mathbb{P}[|\Im \widetilde{E}(\theta)| \ge \delta/4] \le 4\exp(-N\delta^2/128).
\end{equation}
If we choose $\sigma = \mathcal{O}(\log^{1/2}(1/\delta))$ so that the truncation error in~\cref{eqn:truncation-error} is at most $\delta/2$, then $|E(\theta)| \le |\widetilde{E}(\theta)| + \delta/2$. Hence
\begin{equation}
\mathbb{P}[|E(\theta)| \ge \delta] \le \mathbb{P}[|\widetilde{E}(\theta)| \ge \delta/2] \le 4\exp(-N\delta^2/128).
\end{equation}
Therefore, over the finite search space $\Theta$, by the union bound, we have
\begin{equation}\label{eqn:union-bound-on-mesh}
    \mathbb{P}[\max_j |E(\theta_j)| \ge \delta] \le 4J \exp(-N \delta^2/128).
\end{equation}
It suffices to choose $N = \mathcal{O}(\log(J/\eta)/\delta^2)$ so $\mathbb{P}[\max_{\theta \in \Theta}|E(\theta)| \ge \delta] \le \eta$.
This proves~\cref{eqn:success-prob-empirical}.
\end{proof}

Now we turn to prove the performance guarantee of~\cref{algo:qsigs}.

\begin{theorem}\label{thm:qsigs}
    Assume $p_0 = \Tr[\ketbra{v_0}{v_0}\rho] \ge 0.8$, where $\rho$ is an initial state and $\ket{v_0}$ is the ground (right) singular vector of $A$. For a fixed $T>0$ and $\eta > 0$, we choose the following parameters in~\cref{algo:qsigs}:
    \begin{equation}
        N = \mathcal{O}(\log(T/\eta)),\quad q = \sqrt{\frac{\log(10/7)}{2}}.
    \end{equation}
    Let $\theta^*$ be the output of~\cref{algo:qsigs}. With success probability at least $1-\eta$, we have $|\sigma_0 - \theta^*| \le 3q/T$.
\end{theorem}
\begin{proof}
    We set $\delta = 0.01$, and $\sigma = \mathcal{O}(\log^{1/2}(1/\delta))$ is a constant. By~\cref{lem:empirical-error}, with probability at least $1 - \eta$, we have
    \begin{equation}
        \max_{\theta \in \Theta}|E(\theta)| \le 0.01,
    \end{equation}
    by choosing $N = \mathcal{O}(\log(J/\eta)) = \mathcal{O}(\log(T/\eta))$. In what follows, we discuss two cases:
    \begin{enumerate}
        \item If $\theta \in [\sigma_0 - q/T, \sigma_0+q/T]$, we have
    \begin{align}
        |F(\theta)| &\ge \frac{1}{2}\sum^{M-1}_{m=0} p_m \left[e^{-2T^2(\sigma_m-\theta)^2} + e^{-2T^2(\sigma_m + \theta)^2}\right] - |E(\theta)|\\ 
        &\ge \frac{p_0}{2}e^{-2q^2} - |E(\theta)|.\label{eqn:lower-bound-Ftheta}
    \end{align}

    \item If $\theta \notin [\sigma_0 - 3q/T, \sigma_0+3q/T]$, we have
    \begin{align}
        |F(\theta)| &\le \frac{1}{2}\sum^{M-1}_{m=0} p_m \left[e^{-2T^2(\sigma_m-\theta)^2} + e^{-2T^2(\sigma_m + \theta)^2}\right] + |E(\theta)|\\
        &\le p_0 e^{-18q^2} + (1-p_0) + |E(\theta)|.\label{eqn:upper-bound-Ftheta}
    \end{align}
    Indeed, for the $m=0$ term we have $|\theta-\sigma_0| > 3q/T$, and also $\sigma_0+\theta \ge |\theta-\sigma_0| > 3q/T$ because $\sigma_0,\theta \ge 0$. Hence both exponentials in the $m=0$ contribution are at most $e^{-18q^2}$. For the remaining terms, each bracket is at most $2$, so their total contribution is bounded by $1-p_0$.
    \end{enumerate}
    
    Given our choice $q = (\log(10/7)/2)^{1/2} \approx 0.4223$, with probability at least $1-\eta$, we have
    \begin{align*}
        \min_{\theta \in [\sigma_0 - q/T,\sigma_0+q/T]} |F(\theta)| &\ge 0.4e^{-2q^2}-0.01 = 0.27,\\
        \max_{\theta \notin [\sigma_0 - 3q/T,\sigma_0+3q/T]}|F(\theta)| &\le 0.8e^{-18q^2}+0.21 < 0.243.
    \end{align*}
    Since the mesh spacing is $q/T$, there exists at least one grid point in $\Theta$ that lies in $[\sigma_0-q/T,\sigma_0+q/T]$. Therefore the maximizer cannot lie outside $[\sigma_0-3q/T,\sigma_0+3q/T]$, and hence $|\theta^* - \sigma_0| \le 3q/T$ with probability at least $1-\eta$.
\end{proof}

We remark that the condition $p_0 \ge 0.8$ in the previous theorem is chosen for the ease of analysis, and it can be relaxed.
To see this, note that to ensure~\cref{eqn:upper-bound-Ftheta} is upper bounded by~\cref{eqn:lower-bound-Ftheta}, we essentially need $p_0 e^{-\mathcal{O}(q^2)}> (1-p_0)+2\delta$.
As long as $p_0$ satisfies the Sufficiently Dominant Condition~\cite{ding2024quantum}, i.e. $p_0 > 0.5$, we can always pick a $\delta < p_0-1/2$ so there always exists a $q >0$ such that the peak of $|F(\theta)|$ is obtained in the interval $[\sigma-3q/T,\sigma+3q/T]$.
In this case,~\cref{algo:qsigs} still works with a similar asymptotic complexity.

\section{Circuit implementation of the discrete-time Lindbladian dynamics}
\label{append:circuit-implementation}
While several quantum algorithms for simulating continuous-time Lindbladian dynamics have been proposed~\cite{kliesch2011dissipative, cleve2017efficient, childs2017efficient, li2022simulating, ding2024simulating}, they often require complicated algorithmic gadgets that are beyond the capability of near-term quantum devices. Based on the discrete-time Lindbladian dynamics discussed in~\cite{ding2024single}, we propose a simple discrete-time circuit implementation that prepares the ground singular vector of $A_z$. 

Note that the Lindbladian operator $\mathcal{L}$ in~\cref{eqn:lindbladian} has two parts, namely, the coherent part $\mathcal{L}_H[\rho] \coloneqq -i [H, \rho]$ and the dissipative part $\mathcal{L}_K[\rho] \coloneqq K\rho K^\dag - \frac{1}{2}\{K^\dag K, \rho\}$.
To simulate the Lindbladian dynamics for a short time $\tau$, we can use the first-order Trotter formula:
\begin{equation}
\label{eqn:Lindblad-Trotter}
    e^{\mathcal{L}\tau} = e^{\mathcal{L}_K\tau} e^{\mathcal{L}_H \tau}  + \mathcal{O}(\tau^2).
\end{equation}
To approximate the continuous-time evolution over a total time $T$ by repeated first-order Trotter steps, one must choose $\tau$ sufficiently small so that the accumulated Trotter error remains controlled. We will not rely on such a trajectory-approximation regime below. Instead, we regard a single Trotter step itself as a discrete-time quantum channel and use only its fixed-point structure.
Meanwhile, we also notice that the first-order Trotter formula defines a new quantum channel:
\begin{equation}
    \mathcal{N}_\tau[\rho] = e^{\mathcal{L}_K\tau} e^{\mathcal{L}_H \tau}[\rho],
\end{equation}
Because $H$ acts trivially on its ground-state subspace, every density operator supported on that subspace is fixed by $e^{\mathcal{L}_H\tau}$. On the dissipative side, the condition $\hat{f}(\omega)=0$ for $\omega \ge 0$ implies that $K$ annihilates the ground-state subspace: if $\Pi_0$ denotes the ground-state projector of $H$, then $K\Pi_0=0$. Hence $\mathcal{L}_K[\rho]=0$ for every $\rho$ with $\rho=\Pi_0\rho\Pi_0$, so the same holds for $e^{\mathcal{L}_K\tau}$. Therefore every state supported on the ground-state subspace is a fixed point of $\mathcal{N}_\tau$ for any $\tau \ge 0$, including the degenerate case. This motivates iterating $\mathcal{N}_\tau$ as a discrete-time preparation channel. The advantage is that one may take $\tau=\mathcal{O}(1)$ without trying to track the full continuous-time trajectory. The tradeoff is that the resulting convergence rate is a property of the discrete-time channel itself, and need not coincide with the mixing time of the continuous-time Lindbladian.

The coherent part $e^{\mathcal{L}_H \tau}$ can be implemented via standard Hamiltonian simulation. We denote the resulting unitary channel as $\mathcal{U}_\tau[\rho] \coloneqq e^{-i\tau H}\rho e^{i\tau H}$.
The dissipative part can be implemented via Stinespring dilation. Let $\widetilde{K} = \begin{bmatrix} 0 & K^\dag\\K & 0 \end{bmatrix}$ be the Hermitian dilation of the jump opertor $K$, and we define a new quantum channel:
\begin{equation}
    \mathcal{K}_\tau[\rho] = \text{Tr}_a \left( e^{-i \widetilde{K}\sqrt{\tau}}\left[\ketbra{0}{0}_a \otimes \rho\right] e^{i \widetilde{K}\sqrt{\tau}}\right),
\end{equation}
where $a$ stands for the single ancilla qubit introduced in $\widetilde{K}$. 
A key observation is that if we expand $\mathcal{K}[\rho]$ up to order $\mathcal{O}(\tau)$, it matches with the first-order Taylor expansion of the quantum channel $e^{\mathcal{L}_K \tau}$. In other words, for small $\tau$, the channel $\mathcal{K}_\tau[\cdot]$ simulates a Stinespring dilation of $e^{\mathcal{L}_K \tau}$.
For larger $\tau$ (i.e., $\tau = \mathcal{O}(1)$), the channel $\mathcal{K}_\tau[\cdot]$ may deviate from the Lindbladian evolution $e^{\mathcal{L}_K \tau}$, while the ground state of $H$ is always a fixed point of $\mathcal{K}_\tau[\cdot]$, regardless of the choice of $\tau$.

Combining the coherent and dissipative parts of the Lindbladian evolution, we arrive at the following discrete-time dynamics: 
\begin{equation}
    \rho_{k+1} = \Phi_\tau[\rho_k] \coloneqq \mathcal{K}_\tau \circ \mathcal{U}_\tau  [\rho_k],\quad k = 0,1,\dots.
\end{equation}
For any step size $\tau$, the quantum channel $\Phi_\tau[\cdot]$ therefore leaves the ground-state subspace invariant. The quantity $t_k = k\tau$ may still be used as an effective runtime parameter, but for $\tau=\mathcal{O}(1)$ the discrete-time trajectory $\{\rho_k\}$ may not be identified with snapshots of the continuous-time Lindbladian flow.

\section{Mixing time of the Hatano-Nelson model}\label{append:hatano-nelson-mixing}

\subsection{Construction of the Lindbladian operator}

We now consider the Hatano-Nelson model with PBC:
\begin{align}
    H_{\rm HN} = \sum^{n-1}_{j=0} (J+\gamma) \ketbra{j+1}{j} + (J-\gamma) \ketbra{j}{j+1}.
\end{align}
We let $U$ denote the quantum Fourier transform, i.e., $U\ket{j} = \frac{1}{\sqrt{n}}\sum^{n-1}_{k=0} \omega^{jk} \ket{k}$.
Direct calculation shows that $H_{\rm HN}$ can be diagonalized by QFT:
\begin{align}
    U^\dagger H_{\rm HN} U = \sum^{n-1}_{k=0} E_k \ketbra{k}{k},\quad E_k = 2J\cos(2\pi k/n)-2i\gamma\sin(2\pi k/n).
\end{align}
The eigenvalues of $H_{\rm HN}$ are distributed on an ellipse centered at the origin.
At $z = 0$, we form the quadratic Hamiltonian $H_0 = H^\dagger_{\rm HN} H_{\rm HN}$, which is again diagonal in the Fourier basis:
\begin{align}
    H_0 = \sum_k |E_k|^2 \ketbra{\psi_k}{\psi_k},\quad \ket{\psi_k}= U\ket{k}.
\end{align}
Assuming that $0 < \gamma < J = 1$. 
In the thermodynamic limit $n\to \infty$, the ground energy of $H_0$ converges to $4\gamma^2$ (corresponding to the short axis of the ellipse).

Now, we consider two coupling operators: 
\begin{align}
    O_0 = \sum^{n-1}_{k=0} e^{2\pi i k/n} \ketbra{k}{k},\quad O_1 = \sum^{n-1}_{k=0} e^{-2\pi i k/n} \ketbra{k}{k}.
\end{align}
The corresponding jump operators are given by (for $a = 0,1$):
\[K_a = \int f(s) e^{i{H_0}s} O_a e^{-i{H_0}s}~\d s = \sum_{j,k} \hat{f}(\lambda_j-\lambda_k) \bra{\psi_j}O_a\ket{\psi_k} \ketbra{\psi_j}{\psi_k},\]
where $\{(\lambda_j,\ket{\psi_j})\colon j = 0,\dots, n-1\}$ denotes the eigenpairs of $H_0$.
Direct calculation shows that 
\begin{align}
    \bra{\psi_j}O_0\ket{\psi_k} &= \frac{1}{n}\left(\sum_p \omega^{-jp}\bra{p}\right)\left(\sum_{q}e^{2\pi i q/n}\ketbra{q}{q}\right)\left(\sum_r \omega^{kr}\ket{r}\right) = \frac{1}{n} \left(\sum_p \omega^{-jp + kp - p}\right) = \delta_{k,j+1}.
\end{align}
And similarly,
\begin{align}
    \bra{\psi_j}O_1\ket{\psi_k} = \delta_{k,j-1}.    
\end{align} 

For simplicity, we assume $n$ is an integral multiple of $4$, i.e., $n = 4m$. In this case, the (2-fold) ground energy of $H_0$ is obtained at $k = m$ and $k = 3m$:
\begin{align}
    \lambda_m = \lambda_{3m} = 4\gamma^2.
\end{align}
Note that $\lambda_{j}-\lambda_{j+1} < 0$ only when $m \le j \le 2m-1$ and $3m \le j \le 4m-1$ (by identifying $\lambda_{4m} = \lambda_0$). It turns out that 
\begin{align}
    K_0 = U^\dagger \left(\sum^{2m-1}_{j=m} \ketbra{j}{j+1} + \sum^{4m-1}_{j=3m} \ketbra{j}{j+1}\right) U.
\end{align}
Similarly, since $\lambda_j - \lambda_{j-1} < 0$ for $1 \le j \le m$ and $2m+1 \le j \le 3m$, we have
\begin{align}
    K_1 = U^\dagger \left(\sum^{m}_{j=1} \ketbra{j}{j-1} + \sum^{3m}_{j=2m+1} \ketbra{j}{j-1}\right) U .
\end{align}

\begin{figure}[!ht]
    \centering
    \includegraphics[width=0.6\linewidth]{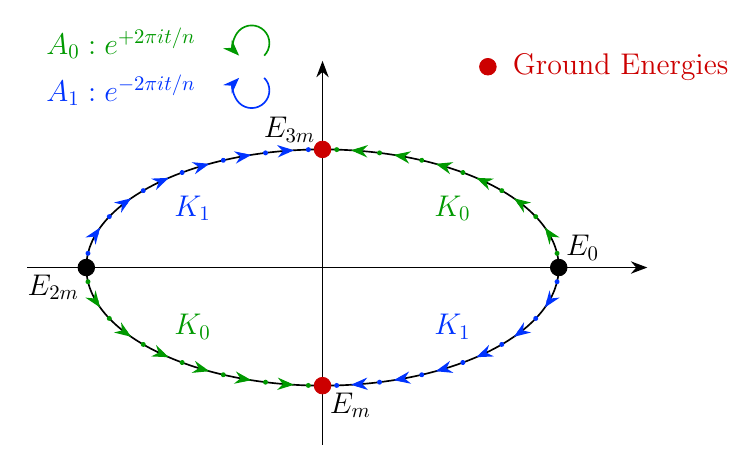}
\caption{Illustration of the spectrum of the Hatano-Nelson model (with $n = 4k$) and the effects of the chosen jump operators $K_0$ and $K_1$.}
\label{fig:hn-pbc}
\end{figure}

By the geometry of the ellipse, $H_0$ can be partitioned into four diagonal blocks (in the Fourier basis):
\begin{align}
    H_0 = U^\dagger \diag(D^{\downarrow}, D^{\uparrow}, D^{\downarrow}, D^{\uparrow}) U, 
\end{align}
with 
\[D^{\downarrow} = \diag(\lambda_0, \lambda_1,\dots, \lambda_{m-1}), D^{\uparrow} = \diag(\lambda_{m}, \lambda_{m-1},\dots, \lambda_1),\]
and $\lambda_k = 2(J^2+\gamma^2)+2(J^2-\gamma^2)\cos(\pi k/m)$ for $k = 0, \dots, m$.

\subsection{Convergence analysis with population dynamics}

Now, we consider the population master equation of the dynamics.
Define a probability vector $\vp(t) = (p_j(t))^{n-1}_{j=0}$ with $p_j(t) \coloneqq \Tr[\ketbra{\psi_j}{\psi_j}\rho(t)]$, where $\{\ket{\psi_j}\}$ are the eigenbasis (i.e., Fourier basis) of $H_0$.
Since $\mathcal{L}^\dagger$ maps every energy basis to a linear combination of neighboring energy eigenbasis, the evolution of $\vp(t)$ can be expressed by the following forward Kolmogorov equation:
\begin{align}
    p'_j(t) = \begin{cases}
        p_{j-1}(t) - p_j(t) & j=1,\dots,m-1,2m+1,\dots,3m-1,\\
        p_{j+1}(t) - p_j(t) & j=m+1,\dots,2m-1,3m+1,\dots,4m-1,\\
        - 2p_j(t) & j = 0, 2m, \\
        p_{j+1} + p_{j-1} & j = m, 3m.
    \end{cases}
\end{align}

To simplify the model, we now introduce a new probability vector $\vq(t)\in \RR^{m+1}$ such that
\begin{align}
    q_j(t) = \begin{cases}
        p_0(t) + p_{2m}(t) & j = 0,\\
        p_j(t) + p_{2m-j}(t) + p_{2m+j}(t) + p_{4m-j}(t) & j = 1,\dots, m-1,\\
        p_m(t) + p_{3m}(t) & j = m.
    \end{cases}
\end{align}
The evolution of $\vq(t)$ is governed by
\begin{align}\label{eq:generator}
    \vq' = L \vq, \quad L = \begin{bmatrix}
        -2 & 0 & 0 & \cdots & 0\\
        2 & -1 & 0 & \ddots & \vdots\\
        \vdots & \ddots & \ddots & \ddots & 0\\
        \vdots & & 1 & -1 & 0\\
        0 & \cdots & \cdots & 1 & 0
    \end{bmatrix}.
\end{align}
In probabilistic terms, from any state $j \in \{1,\dots,m-1\}$ the chain waits an exponential time of mean $1$ and then jumps to $j+1$. The only exception is the state $j=0$, which waits an exponental time of mean $1/2$ before jumping to $j=1$.
The state $j=m$ is absorbing, so the chain has a unique stationary distribution $\pi = \ve_m$.
Note that the holding time on $j=0$ is, on average, shorter than the rest of the states $1 \le j \le m-1$. The mixing time of the chain only becomes longer if we replace the edge entries $\pm 2$ to $\pm 1$, because this increases the mean holding time at $j=0$. 
In what follows, to simplify the analysis and obtain an analytical bound on the mixing time, we instead consider the modified population dynamics that replace the $\pm 2$ entries in the generator matrix $L$ with $\pm 1$, respectively.
By the Markov property, for this modified chain the holding times on the states $j\in \{0,\dots, m-1\}$ are independent $\mathrm{Exp}(1)$ random variables.

Let $\mu \in \RR^{m+1}$ be any probability vector, then
\begin{equation}\label{eq:tv-abs}
\norm{\mu-\pi}_{\mathrm{TV}} = 1-\mu_m.
\end{equation}

We now denote $\mc{X} = \{0,\dots, m\}$ as the state space.
For $x \in \mc{X}$, let $\vp^{(x)}(t)$ denote the law vector of $X_t$ under the initial condition $X_0=x$. In particular, $\vp^{(x)}(0)=\ve_x$. Define the absorption time 
\begin{equation}\label{eq:abs-time}
T_x := \inf\{t \ge 0 : X_t = m \mid X_0=x\}.
\end{equation}
Since state $m$ is absorbing, the event $\{X_t \ne m\}$ is the same as $\{T_x > t\}$, and therefore
\begin{equation}\label{eq:tv-survival}
\norm{\vp^{(x)}(t)-\pi}_{\mathrm{TV}} = \mathbb{P}(X_t \ne m \mid X_0=x) = \mathbb{P}(T_x > t).
\end{equation}
Starting from $x \le m-1$, the modified chain must execute exactly $m-x$ jumps to reach state $m$. Since each pre-absorption holding time is an independent $\mathrm{Exp}(1)$ random variable, $T_x$ is distributed as a sum of $m-x$ i.i.d.\ $\mathrm{Exp}(1)$ variables, i.e.,
\begin{equation}\label{eq:Tx-sum}
T_x := \sum_{k=x}^{m-1} E_k,
\end{equation}
where $(E_k)_{k \ge 1}$ represents the holding time on state $k$. Define partial sums $S_0:=0$ and
\begin{equation}\label{eq:Sn}
S_n := \sum^{m-1}_{k=m-n} E_k \qquad (n \ge 1).
\end{equation}
Define the counting process
\begin{equation}\label{eq:Nt}
N_t := \max\{n \in \NN: S_n \le t\} \qquad (t \ge 0).
\end{equation}
Then for each fixed $t \ge 0$, the random variable $N_t$ satisfies
\begin{equation}\label{eq:Nt-poisson}
\mathbb{P}(N_t = n) = e^{-t}\frac{t^n}{n!}, \qquad n \in \NN \cup \{0\}.
\end{equation}
That is, $N_t \sim \mathrm{Poisson}(t)$. For $x \in \mc{X}$ and $t \ge 0$,
\begin{equation}\label{eq:abs-prob-poisson}
\mathbb{P}(X_t = m \mid X_0=x) = \mathbb{P}(N_t \ge m-x).
\end{equation}
Consequently,
\begin{equation}\label{eq:tv-poisson-2}
\norm{\vp^{(x)}(t)-\pi}_{\mathrm{TV}}
=\mathbb{P}(T_x>t)
=\mathbb{P}(S_{m-x}>t)
=\mathbb{P}(N_t \le m-x-1).
\end{equation}
From the worst case choice $x=0$, we have
\begin{equation}\label{eq:tv-explicit}
\norm{\vp^{(0)}(t)-\pi}_{\mathrm{TV}}
=\mathbb{P}(N_t \le m-1)
=e^{-t}\sum_{k=0}^{m-1}\frac{t^k}{k!}.
\end{equation}
Therefore, the mixing time of the chain is
\begin{equation}\label{eq:tmix}
t_{\mathrm{mix}}(\epsilon) := \inf\Bigl\{t \ge 0 : \sup_{x \in \mc{X}} \norm{\vp^{(x)}(t)-\pi}_{\mathrm{TV}} \le \epsilon \Bigr\}= \inf\Bigl\{t \ge 0 : \mathbb{P}(N_t \le m-1) \le \epsilon \Bigr\}.
\end{equation}
Recall the standard Chernoff estimates for Poisson tails: given $Z \sim \mathrm{Poisson}(\lambda)$ for some $\lambda > 0$,
\begin{align}
    \mathbb{P}(Z \le \lambda-a) &\le \exp\Bigl(-\frac{a^2}{2\lambda}\Bigr),\quad a \in [0, \lambda],\label{eqn:chernoff-1}
\end{align}
Assume that $L_\epsilon := \log(1/\epsilon) \le m$. We set $t := m + 3\sqrt{mL_\epsilon}$ and $a := t-(m-1)$. Note that $t > m$ so $a \ge 0$. Then using \cref{eqn:chernoff-1} with $N_t \sim \mathrm{Poisson}(t)$, we have
\begin{equation}\label{eq:upper-bound}
\mathbb{P}(N_t \le m-1) \le \exp\Bigl(-\frac{a^2}{2t}\Bigr),\quad a = 3\sqrt{mL_\epsilon}+1.
\end{equation}
Since $a^2 \ge 9mL_\epsilon$ and $L_\epsilon \le m$ implies $t \le 4m$, the exponent in~\cref{eq:upper-bound} satisfies
\[
    \frac{a^2}{2t} \ge \frac{9mL_\epsilon}{8m} > L_\epsilon,
\]
which implies that $\mathbb{P}(N_t \le m-1) \le e^{-L_\epsilon}=\epsilon$. This gives an upper bound on the mixing time of the modified population dynamics in the regime $L_\epsilon \le m$:
\begin{equation}
    t_{\mathrm{mix}}(\epsilon) \le m + 3\sqrt{m \log(1/\epsilon)}.
\end{equation}

\begin{figure}
    \centering
    \includegraphics[width=0.5\linewidth]{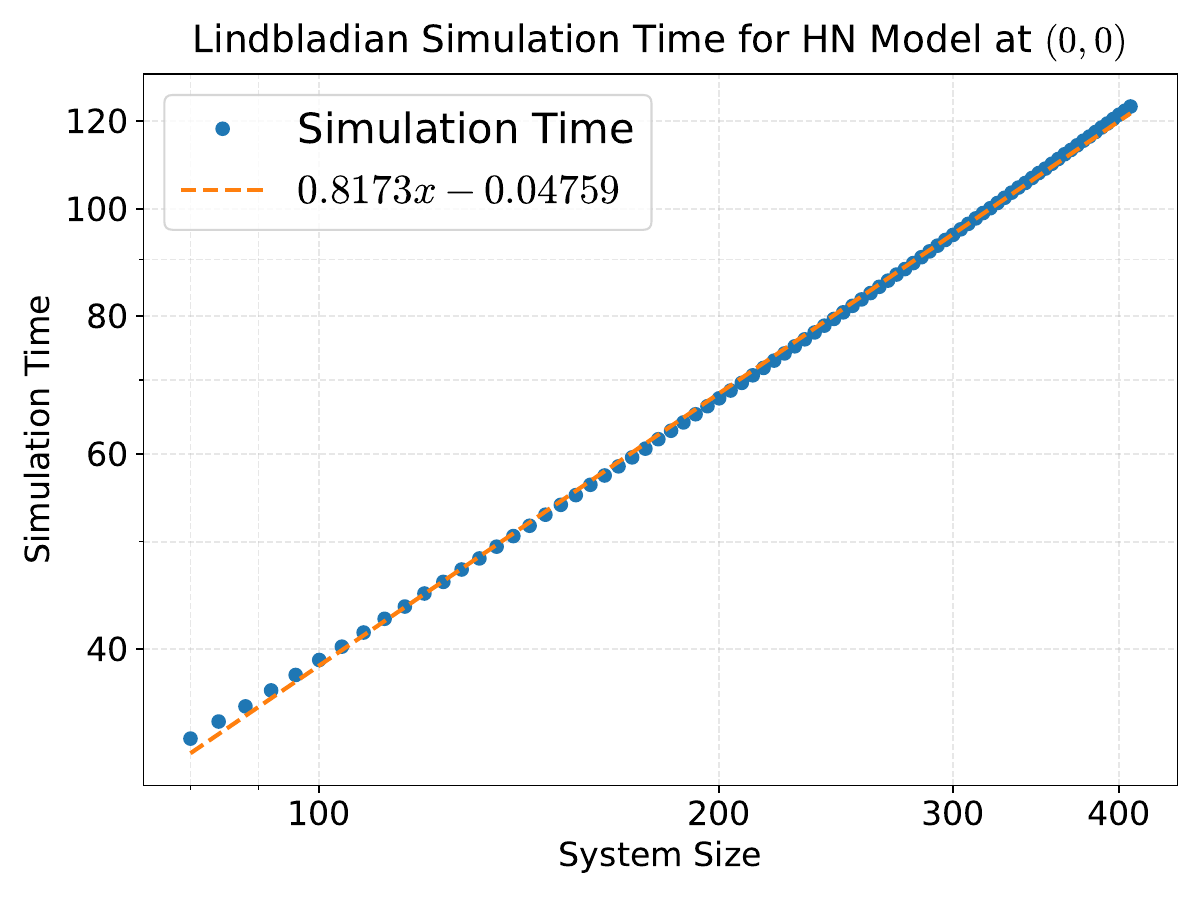}
    \caption{Scaling behavior of the Lindbladian simulation time for the Hatano-Nelson model at $z = 0$ as system size increases, plotted on a log-log scale. The step size is set as $\tau = 0.1$ and the convergence threshold is $\left|\text{Tr}(H_{\text{HN}}^\dag H_{\text{HN}}\rho(t)) - E_0(H_{\text{HN}}^\dag H_{\text{HN}})\right| \leq 10^{-3}$. The overlaid first-degree polynomial fit demonstrates that the simulation time exhibits (sub-)linear growth.}
    \label{fig:hatano-nelson-pbc-zero-mixing-time}
\end{figure}

\subsection{Hatano-Nelson model under open boundary conditions}
For completeness, we provide numerical evidence for the Lindbladian simulation time of the Hatano-Nelson model under open boundary conditions.
In this numerical test, we take
\begin{equation}
  P=\sum_{j=0}^{n-1} x_j \ketbra{j}{j},
  \qquad
  x_j = -1 + 2j/n,
\end{equation}
as a coarse-grained position observable, together with the reflection operator
\begin{equation}
    R = \sum_{j=0}^{n-1} \ketbra{n-1-j}{j}
\end{equation}
as coupling operators. In the simulation, we scale these coupling operators by a constant factor $10$. The heat map is shown in \cref{fig:hatano-nelson-obc}.
For the OBC Hatano--Nelson model, the observed convergence time (bounded by $T = 100$) is slightly longer than that in the PBC case (bounded by $T = 30$). This behavior is consistent with the increased spectral sensitivity under OBC, and better choices of coupling operators may further accelerate the mixing.

\begin{figure}
    \centering
    \includegraphics[width=0.5\linewidth]{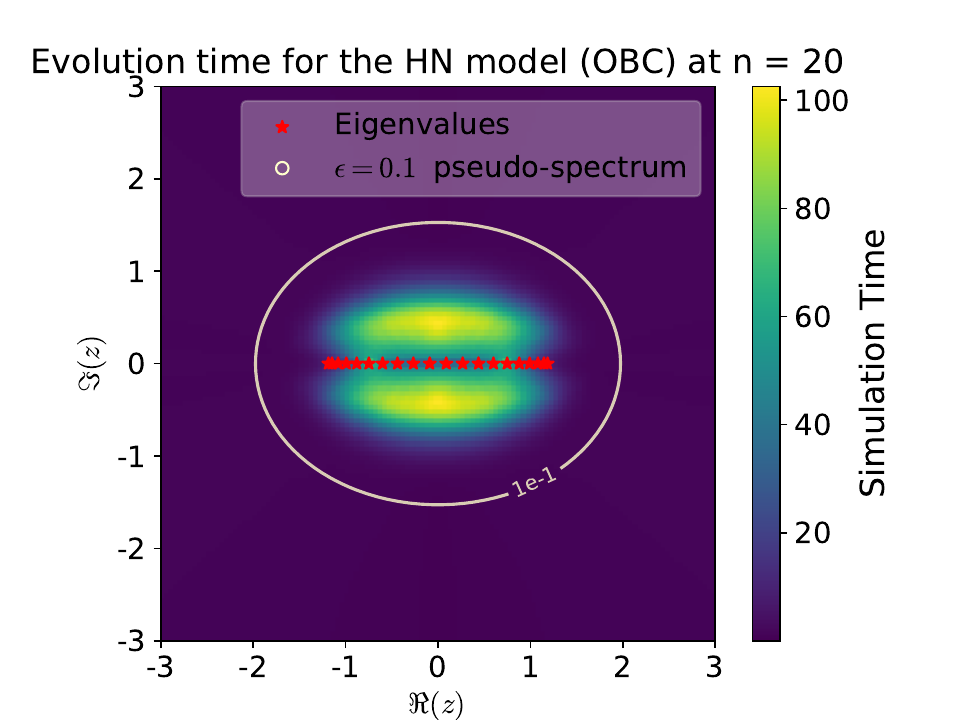}
    \caption{The $0.1$-pseudospectrum of the Hatano-Nelson model under OBC, and the mixing time of the Lindbladian dynamics for preparing the ground singular vector of $H_{\rm HN}-zI$.}
    \label{fig:hatano-nelson-obc}
\end{figure}

\section{Details on the non-Hermitian qubit system and IonQ implementation}\label{append:numerical}

\subsection{Spectral properties of the 1-qubit model}
Recall the 2-by-2 matrix with $g>0$:
\begin{equation}
    H(g) = X - i g Z = 
    \begin{bmatrix}
        -i g & 1 \\
        1 & i g
    \end{bmatrix}.
\end{equation}
The eigenvalues of the matrix are: $\lambda_{\pm} = \pm \sqrt{1-g^2}$. They are purely real for $g < 1$, and purely imaginary for $g > 1$. When $g = 1$,  the two eigenvalues coalesce, indicating an order-2 exceptional point (EP). 

We define $H_z \coloneqq (H(g) - zI)^\dagger (H(g) - zI)$. 
When $g = 1$, the eigenvalues of $H_z$ are $2 + |z|^2 \pm 2\sqrt{|z|^2+1}$, meaning that $\sigma_0(H(1)-zI) = \sqrt{|z|^2+1}-1$. So, the $\epsilon$-pseudospectrum of $H(1)$ is a disc centered at the origin with radius $r=\sqrt{2\epsilon+\epsilon^2}$:
\[\Lambda_\epsilon = \left\{z\in \mathbb{C}\colon \sigma_0(H(1)-zI) \le \epsilon\right\} = \left\{z\in \mathbb{C}\colon |z| \le \sqrt{\epsilon^2+2\epsilon}\right\}.\]
The radius of the pseudo-spectral cloud scales as $\mathcal{O}(\sqrt{\epsilon})$ for small $\epsilon$, which is a signature of an order-2 EP.

For $g \neq 1$, a direct calculation gives $\sigma_0(H(g)) = |1-g|$, so $0 \in \Lambda_\epsilon(H(g))$ if and only if $\epsilon \ge |1-g|$. Thus, the two eigenvalues remain separated in the pseudospectrum for tolerances below this scale, and they merge only when $\epsilon$ reaches order $|1-g|$.
\begin{figure}
    \centering
    \includegraphics[width=0.5\linewidth]{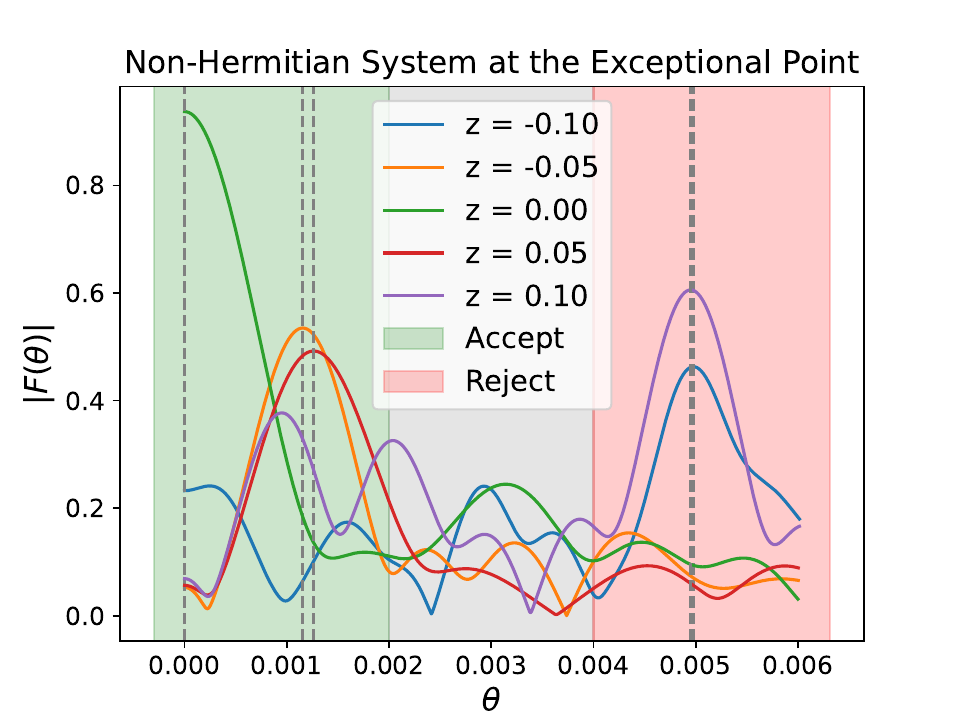}
    \caption{The filter functions in QSIGS constructed using the complete measurement data from IonQ for $z = -0.1, -0.05, 0, 0.05$ and $0.1$. The settings are similar to that in~\cref{fig:NH}(c). With fixed $\epsilon = 0.002$, we show $z = -0.05, 0, 0.05$ lie in the $\epsilon$-pseudospectrum and $z = -0.1$ and $z = 0.1$ do not.}
    \label{fig:nh-all-ionq}
\end{figure}

\subsection{Implementation of QSIGS}\label{append:implement-qsigs}

We use block-encodings to implement the quantum singular value transformation (QSVT) in QSIGS.
Note that a block-encoding can be easily constructed if the singular value decomposition is known~\cite{lin2022lecture}. For a general matrix $M$ ($\|M\| \leq 1$) along with its singular value decomposition $M = W \Sigma V^\dag$,
its block-encoding can be constructed as
\begin{equation}\label{eqn:block-encoding-svd}
    U = 
    \begin{bmatrix}
        M & W \sqrt{I - \Sigma^2}\\
        \sqrt{I - \Sigma^2} V^\dag & -\Sigma
    \end{bmatrix}.
\end{equation}
To execute the experiments, we start from a general matrix $A$ and a complex number $z$. 
We assume the matrix $A_z = A - zI$ has a singular value decomposition as $A_z = W \Sigma V^\dag$. We denote $\widetilde{\Sigma} = \Sigma/\sigma_{\max}$ where $\sigma_{\max}$ is the largest singular value of $A_z$. In order to satisfy the normalization constraint, we consider $\widetilde{A_z} := W \widetilde{\Sigma} V$. Then,~\cref{eqn:block-encoding-svd} gives us a block-encoding of $\widetilde{A_z}$ by introducing one ancilla qubit.
Here in the experiments, for each chosen $z$, we recorded the normalization constants $\sigma_{\max}$ and rescaled the $\theta$'s 
for getting a precise estimation of the ground singular values $\sigma_0(A_z)$.

For the $\sin$ transformation, the same procedure is applied to $W \sin\left(t\widetilde{\Sigma}\right) V^\dag$, which yields a block-encoding of $\sinsv(t\widetilde{A_z})$.
In our real-hardware experiment, we decompose the unitary block-encodings into elementary gates using~$\texttt{qiskit.transpile}$~\cite{qiskit2024} with optimization level 3.
The transpiled circuit (together with dissipative ground state preparation) uses $\sim$ 130 single-qubit gates and $\sim$ 20 $2$-qubit gates.

We used QSIGS to estimate the ground singular value of $A_z$. For the system $\widehat{H}(1)$, we tested 5 different $z$'s, namely $z = \pm0.1, \pm0.05, 0$.
We choose the distribution function $a(t) = \frac{1}{\sqrt{2\pi}T}\exp\left(\frac{-t^2}{2T^2}\right)$ with $T = 2000$. 

For each choice of $z$, we sampled $25$ time points $\{t_n\}$ according to $a(t)$ and then executed $25$ corresponding circuits on the IonQ device. For each circuit, we ran $2000$ shots; thus, for each choice of $z$, the total number of samples $\{Z_n\}$ is $5\times10^4$.
It is worth noting that this procedure differs slightly from~\cref{algo:qsigs}, where each $t_n$ produces a single sample $Z_n$. Here we reuse the same quantum circuit to generate multiple samples per $t_n$, which reduces the total number of circuit compilations and keeps the experiment budget manageable.
With these samples, we calculate the filter function $F(\theta)$ and seek for the maximizer of $|F(\theta)|$ (denoted as $\theta^*$) on $5\times 10^5$ uniformly sampled nodes on $[0, 6\times10^{-3}]$.
Assuming that the dissipative protocol prepares an initial state $\rho_z$ such that $p_0 = \Tr[\ketbra{\psi_z}{\psi_z}\rho_z]\ge 0.8$ (which is verified by numerical simulation),~\cref{thm:qsigs} suggests that with high probability, the ground singular value $\sigma_0$ falls into the confidence interval:
\begin{equation}\label{eqn:confidence-interval}
    \left[\theta^* - \frac{3q}{T},\theta^* + \frac{3q}{T}\right],\quad \frac{3q}{T} = \frac{3(\log(10/7)/2)^{1/2}}{2000} \approx 6\times 10^{-4}.
\end{equation}
In~\cref{table:numerical-result-tab}, we list the maximizers of $|F(\theta)|$ for each choice of $z$, together with the ground truth singular value $\sigma_0$ and the absolute error $|\theta^*-\sigma_0|$.
\begin{table}[ht!]
\centering
\scalebox{0.9}{
\begin{tabular}{|c|c|c|c|}
\hline
\textbf{$z$ values} & \textbf{Empirical Maximizer $\theta^*$} & \textbf{Ground singular value $\sigma_0$ (ground truth)} & \textbf{Difference $|\theta^* - \sigma_0|$~}\\
\hline
-0.1 & $4.97\times 10^{-3}$  & $4.99\times 10^{-3}$ & $2\times 10^{-5}$\\
-0.05 & $1.15\times 10^{-3}$  & $1.25\times 10^{-3}$ & $1\times 10^{-4}$\\
0 & $0$  & $0$ & $0$\\
0.05 & $1.26\times 10^{-3}$  & $1.25\times 10^{-3}$ & $1\times 10^{-5}$\\
0.1 & $4.97\times 10^{-3}$  & $4.99\times 10^{-3}$ & $2\times 10^{-5}$\\
\hline
\end{tabular}
}
\caption{Ground singular value estimate via QSIGS, with data collected on IonQ Forte.}
\label{table:numerical-result-tab}
\end{table}
The results indicate that, for each $z$, the interval $[\theta^*-3q/T,\theta^*+3q/T]$ lies entirely in either the ``accept'' or ``reject'' regime shown in~\cref{fig:nh-all-ionq}.
The last column of~\cref{table:numerical-result-tab} shows the estimation error (compared to the true ground singular value $\sigma_0$), which is always smaller than the reference scale $3q/T\approx 6\times 10^{-4}$.

\subsection{Implementation of the discrete-time dynamics}\label{append:implement-dsp}

To iteratively apply the quantum channel $\Phi_\tau[\cdot]$ to an initial state, we need to trace out the ancilla qubit in each iteration and reset the ancilla qubit to $\ket{0}$. This operation is referred to as a $\texttt{reset}$ gate in modern quantum hardware. Unfortunately, the reset operation is currently not supported by most commercial quantum computers (including IonQ) due to the requirement of mid-circuit measurement.
To circumvent this difficulty, it is possible to postpone all the mid-circuit measurements until the final measurement. 
We now explain the procedures in detail.

Suppose that we want to apply the quantum channel $\Phi_\tau[\cdot]$ twice to an initial state $\rho_0$. The resulting quantum operations are expressed by (for simplicity, we omit the unitary channel $\mathcal{U}_\tau$ as it does not affect the procedure):
\begin{equation}
    \Phi^{(2)}_\tau[\rho_0] = \Tr_b\left[e^{-i\sqrt{\tau}\widetilde{K}_b} \left(\ketbra{0}{0}_b \otimes \Tr_a\left[e^{-i\sqrt{\tau}\widetilde{K}_a} \left(\ketbra{0}{0}_a \otimes \rho_0\right) e^{i\sqrt{\tau}\widetilde{K}_a}\right]\right)e^{i\sqrt{\tau}\widetilde{K}_b}\right],
\end{equation}
where $a$, $b$ represent two independent ancilla qubits, and $\widetilde{K}_a$, $\widetilde{K}_b$ are Hermitian dilations of $K$ using ancilla $a$, $b$, respectively.
Note that the outer unitary $e^{-i\sqrt{\tau}\widetilde{K}_b}$ does not act on the ancilla qubit $a$, so the partial trace $\Tr_a[\cdot]$ can be interchanged with the outer unitary operation. It follows that we can introduce \textit{two} ancilla qubits for two iterations and trace them out at the end of the quantum circuit:
\begin{equation}
    \Phi^{(2)}_\tau[\rho_0] = \Tr_{ab}\left[e^{-i\sqrt{\tau}\widetilde{K}_b} e^{-i\sqrt{\tau}\widetilde{K}_a} \left(\ketbra{00}{00}_{ab} \otimes \rho_0\right) e^{i\sqrt{\tau}\widetilde{K}_a}e^{i\sqrt{\tau}\widetilde{K}_b}\right],
\end{equation}

In general, to apply $N$ iterations of $\Phi_\tau[\cdot]$, we can introduce $N$ ancilla qubits (indexed by $a_k$ for $1 \le k \le N$) and perform a final measurement. If the quantum channel $\Phi_\tau[\cdot]$ has a unique stationary state and a positive spectral gap, for any initial state such as $\rho_0 = \ketbra{0}{0}$, the final state $\Phi^{(N)}[\rho_0]$ should be close to the ground state of $H$. In other words, the final state in the quantum circuit is essentially a product state:
\begin{equation}
    \prod_{k=N}^{1} e^{-i \widetilde{K}_{a_k}\sqrt{\tau}}
    \left(
    \ket{0}_{a_N\cdots a_1}\otimes \ket{0}
    \right)
    \approx \ket{\text{Garbage}}_{a_N \cdots a_1} \otimes \ket{\lambda_0}.
\end{equation}
By discarding all measurement results in the ancilla register, we recover the (approximate) ground state of $H$.

To physically implement such a protocol on the IonQ machine, for each step, one needs to introduce one more
ancilla and apply the unitary $e^{-i \widetilde{K}\sqrt{\tau}}$ on both the target register and the ancilla.
The trace-out process is then replaced by discarding all previous ancilla qubits. 
Note that the ancillas are idling in the system rather than being reset back to a clean $\ket{0}$ state, they are subject to ongoing decoherence and can cause unwanted errors on the target qubits.
The characterization of such errors is beyond the scope of this paper, while the experiment results (see~\cref{table:numerical-result-tab}) suggest that our protocol is robust to machine decoherence when implemented on the IonQ device.
In~\cref{fig:1-qubit-eig-circuit}, we illustrate the quantum circuit that implements the quantum channel $\Phi_\tau[\cdot]$ five times, followed by a quantum singular value estimation via QSIGS. 

\begin{figure}[!ht]
    \centering
    \includegraphics[width=0.7\linewidth]{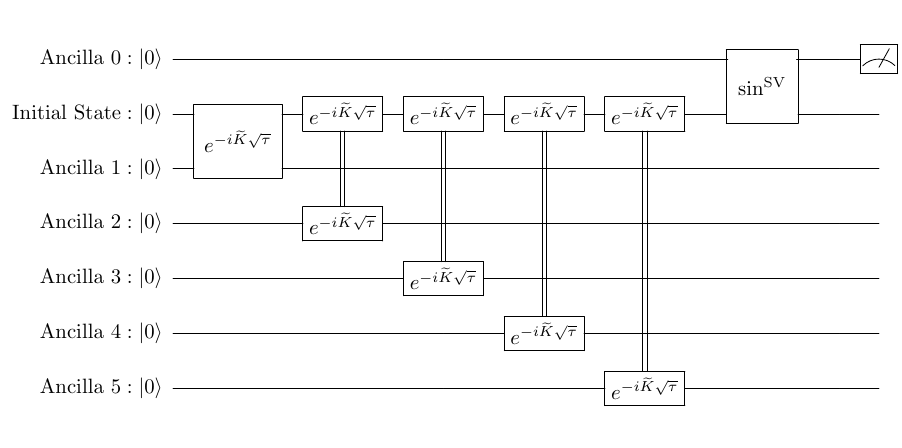}
    \caption{A circuit implementation example for an 1-qubit system.}
    \label{fig:1-qubit-eig-circuit}
\end{figure}

\end{document}